\documentclass[12pt,a4paper]{article}

\usepackage[margin=2.5cm]{geometry}
\usepackage{setspace}
\usepackage{graphicx}
\usepackage{multicol,multirow}
\usepackage{amsmath,amssymb,amsfonts}
\usepackage{amsthm}
\usepackage{rotating}
\usepackage{appendix}
\usepackage{ifpdf}
\usepackage[T1]{fontenc}
\usepackage{newtxtext}
\usepackage{newtxmath}
\usepackage{textcomp}
\usepackage{xcolor}
\usepackage{bm}
\usepackage{booktabs}
\usepackage{enumitem}
\usepackage{array}
\usepackage{etoolbox}
\makeatletter 
\apptocmd{\@xfloat}{\small}{}{\PackageWarning{main}{Could not patch floats}}
\makeatother
\usepackage[style=apa,backend=biber,natbib]{biblatex}
\usepackage[colorlinks,allcolors=blue]{hyperref}

\theoremstyle{definition}
\newtheorem{theorem}{Theorem}
\newtheorem{definition}{Definition}
\numberwithin{equation}{section}

\title{A Hybrid NUTS-Gibbs Sampler with State Space Marginalization for Estimation of Dynamic Structural Equation Models with Binomial Outcomes}
\author{%
  \textbf{Øystein Sørensen}\textsuperscript{1}\quad
  \textbf{Ethan M. McCormick}\textsuperscript{2,3}\\[1ex]
  \footnotesize\textsuperscript{1}\textit{Department of Psychology, University of Oslo, Oslo, Norway}\\
  \footnotesize\textsuperscript{2}\textit{Educational Statistics and Research Methods, School of Education, University of Delaware, Newark, Delaware, USA}\\
  \footnotesize\textsuperscript{3}\textit{Methodology \& Statistics Department, Institute of Psychology, Leiden University, Leiden, Netherlands}}
\date{}

\newcommand{\startsupplement}[1]{%
  \clearpage
  \phantomsection
  \addcontentsline{toc}{section}{Supplementary Material #1}%
  \begin{center}\Large\bfseries Supplementary Material #1\end{center}
  \setcounter{section}{0}\setcounter{equation}{0}\setcounter{figure}{0}%
  \setcounter{table}{0}\setcounter{theorem}{0}\setcounter{definition}{0}%
  \renewcommand{\thesection}{S#1.\arabic{section}}%
  \renewcommand{\theequation}{S#1.\arabic{equation}}%
  \renewcommand{\thefigure}{S#1.\arabic{figure}}%
  \renewcommand{\thetable}{S#1.\arabic{table}}%
  \renewcommand{\thetheorem}{S#1.\arabic{theorem}}%
  \renewcommand{\thedefinition}{S#1.\arabic{definition}}%
  \renewcommand{\theHsection}{supp#1.\arabic{section}}%
  \renewcommand{\theHequation}{supp#1.\arabic{equation}}%
  \renewcommand{\theHfigure}{supp#1.\arabic{figure}}%
  \renewcommand{\theHtable}{supp#1.\arabic{table}}%
  \renewcommand{\theHtheorem}{supp#1.\arabic{theorem}}%
  \renewcommand{\theHdefinition}{supp#1.\arabic{definition}}%
}

\begin{document}

\maketitle

\begin{abstract}
Dynamic structural equation models (DSEMs) are widely used for analyzing intensive longitudinal data (ILD). Although many ILD have categorical (Bernoulli or binomially distributed) responses, currently available algorithms for estimating DSEMs are limited to using the probit link and the Bernoulli distribution, i.e., the response can only be success or failure for a single trial and not the number of successes in multiple trials. They also scale poorly with increasing model complexity and sample size. Here, we present a hybrid sampler---alternating between one step of the No-U-Turn Sampler (NUTS) and one Gibbs step---which solves both of these problems: the Gibbs step naturally handles Pólya-Gamma distributed latent variables arising from binomially distributed responses with a logit link, and the NUTS step utilizes a Kalman filter to exactly marginalize over latent states, making models with dynamics defined in terms of latent traits more computationally feasible. We demonstrate the efficiency of the proposed sampler in simulation experiments, and also study conditions under which it fails. We finally illustrate its use in an example application of predicting panic attacks.
\end{abstract}


\section{Introduction}

Dynamic structural equation modeling (DSEM) \citep{asparouhovDynamicStructuralEquation2018} has become highly popular for analysis of intensive longitudinal data, in which each participant is measured at a high number of timepoints over a relatively short time period \citep{hamakerFrontiersModelingIntensive2018}. Discrete outcomes are common in such data \citep{castro-alvarezManyReliabilitiesPsychological2025}, and DSEMs with binary \citep{mcneishDynamicStructuralEquation2024} or ordinal \citep{mcneishExploringHowMany2025} responses can be analyzed in \textsc{Mplus} \citep{mplus} using a probit link and a latent response formulation \citep{albertBayesianAnalysisBinary1993}.

\citet{asparouhovDynamicStructuralEquation2018} proposed a Metropolis-within-Gibbs sampler for estimating DSEMs, which mainly consists of Gibbs sampling, interspersed with a Metropolis step for updating participant-specific variance components. In the relatively common case in which variances do not vary between participants, it reduces to a pure Gibbs sampler. This algorithm---which is implemented in \textsc{Mplus}---suffers from a computational bottleneck. Letting $N$ denote the number of participants and $T$ denote the number of timepoints, it requires sampling $\mathcal{O}(N \cdot T)$ parameters in each iteration. For conditionally Gaussian responses, \citet{sorensenEfficientBayesianEstimation2026} showed how the number of parameters to sample could be reduced to $\mathcal{O}(N+T)$ by exactly integrating over latent states with a Kalman filter \citep{kalmanNewApproachLinear1960}. However, the method of \citet{sorensenEfficientBayesianEstimation2026} does not work with non-Gaussian responses, as the Kalman filter is no longer exact. Still, given the high number of participants and timepoints in datasets typically found in applications, developing scalable algorithms is critical for enabling the DSEM framework to be fully utilized without forcing practitioners to use sub-optimal models for the sake of computational feasibility alone, e.g., by using sum scores \citep{mcneishThinkingTwiceSum2020} or treating discrete data as continuous \citep{mcneishExploringHowMany2025}. For example, the 60 datasets available in the \textsc{openESM} database \citep{siepeIntroducingOpenESMDatabase2025} at the time of writing, contain on average $N=269$ participants measured at $T=104$ timepoints. Estimating DSEMs with data of this size using currently available algorithms is highly challenging. 

In addition to the computational scaling issue, Gibbs samplers suffer from two further limitations. Firstly, Gibbs sampling requires conditionally conjugate prior distributions, and this severely limits both the ability of practitioners to specify realistic priors based on domain knowledge as recommended in a proper Bayesian workflow \citep{schadPrincipledBayesianWorkflow2020}, and makes it challenging to extend the DSEM framework, e.g., by allowing asymmetric threshold dynamics \citep{schaafUncoveringAsymmetricTemporal2025} or nonlinear effects via regression splines \citep{sorensenModelingCyclesTrends2025}. Secondly, the number of parameters in a DSEM is typically very high in common use cases. For example, a single latent trait in $N=50$ participants over $T=50$ timepoints contributes $2{,}500$ latent trait parameters, in addition to all participant-specific and population-level parameters. Furthermore, these parameters are typically highly correlated. It is well-known that Gibbs sampling tends to be inefficient in such cases \citep{parkImprovingGibbsSampler2022}, producing Markov chains with high autocorrelation. In contrast, Hamiltonian Monte Carlo (HMC) \citep{Neal2011} is particularly efficient for exactly these kinds of challenging posteriors \citep{betancourtConceptualIntroductionHamiltonian2018}. 

In this paper, we aim to improve all the aforementioned issues of the DSEM estimation algorithm of \citet{asparouhovDynamicStructuralEquation2018}, focusing on the case of binomially distributed responses. To this end, we propose a hybrid NUTS-Gibbs algorithm. In the NUTS step, the sampler uses HMC through the No-U-Turn Sampler (NUTS) \citep{hoffmanNoUTurnSamplerAdaptively2014} combined with marginalization of within-level latent states through a Kalman filter exactly as in \citet{sorensenEfficientBayesianEstimation2026}, but conditional on latent responses. In the Gibbs step, latent responses are then sampled conditionally on all other model parameters. The number of parameters to be sampled in each iteration is still $\mathcal{O}(N \cdot T)$, but the computationally demanding NUTS step samples only $\mathcal{O}(N+T)$ parameters while the $\mathcal{O}(N \cdot T)$ latent responses are sampled in a single highly parallelizable Gibbs step. This way, we exploit NUTS' ability to efficiently explore high-dimensional parameter spaces with the speed of the parallelized Gibbs sampler. 

Our proposed algorithm supports both binomially distributed data with the logit link using Pólya-Gamma distributed latent variables \citep{polsonBayesianInferenceLogistic2013} and the probit link using truncated normal latent variables \citep{albertBayesianAnalysisBinary1993}. To the best of our knowledge, estimation of DSEMs with binomial data using the Pólya-Gamma distribution has not been discussed in the literature, but particularly with binomial observations based on multiple trials, this is a more efficient approach than using truncated normals, which requires either new latent variables or expanding the data into Bernoulli trials. However, such responses are common, e.g., in data from educational technology \citep{juddTrainingSpatialCognition2021,settlesTrainableSpacedRepetition2016} and cognitive testing \citep{paolilloCharacterizingPerformanceSuite2024,schmitter-edgecombeCapturingCognitiveCapacity2025,sliwinskiReliabilityValidityAmbulatory2018} and model expansions which accommodate these data types are warranted.

The paper proceeds as follows: In Section \ref{sec:DSEM} we formally define the DSEM framework with discrete responses, in Section \ref{sec:StateSpaceFormulation} we review the state space formulation of \citet{sorensenEfficientBayesianEstimation2026} and show how it applies with discrete responses, and in Section \ref{sec:NUTSGibbs} we define our hybrid NUTS-Gibbs sampler. In Section \ref{sec:Simulations} we show simulation experiments for various DSEMs with binomial responses, comparing the NUTS-Gibbs sampler to alternative algorithms and in Section \ref{sec:ApplicationExample} we show an example application where we analyze an ecological momentary assessment dataset with a mix of binomial and continuous observations. In Section \ref{sec:Limitation} we demonstrate why the proposed hybrid sampler is inefficient for ordinal responses and discuss the challenges it faces with negative binomially distributed responses, and in Section \ref{sec:Discussion} we discuss a number of interesting opportunities for further improvements in DSEM estimation.

\section{Dynamic Structural Equation Models with Discrete Responses}
\label{sec:DSEM}

We here show an extension of DSEM for handling discrete responses using either a logit or a probit link. While \citet{asparouhovDynamicStructuralEquation2018} pointed out how this can be done with a probit link, and \citet{mcneishDynamicStructuralEquation2024,mcneishExploringHowMany2025} show several examples, to the best of our knowledge the complete mathematical formulation---which differs from that used with purely Gaussian responses---has not been presented in the literature.

Consider a $U$-dimensional response vector $\bm{y}_{it}$ which is measured in $N$ participants $(i=1,\dots,N)$ at $T$ timepoints ($t=1,\dots,T$). Let the observation vector be drawn from a family of distributions $\mathcal{F}$ parametrized by an expected value $\bm{\mu}_{it}$ and a (possibly empty) set of ancillary parameters $\bm{\theta}_{it}$, $\bm{y}_{it} \sim \mathcal{F}(\bm{\mu}_{it}, \bm{\theta}_{it})$. The linear predictor, which we denote by $\bm{y}_{it}^{*}$, is connected to $\bm{\mu}_{it}$ through $\bm{\mu}_{it} = H(\bm{y}_{it}^{*})$, where $H(\cdot)$ is either the identity ($H(y) = y$), probit ($H(y) = \Phi(y)$; where $\Phi(\cdot)$ is the cumulative distribution function of the standard normal distribution \citep[Ch. 2]{skrondalGeneralizedLatentVariable2004}), or logit ($H(y) = [1+\exp(-y)]^{-1}$) inverse link function. Consequently, the choice of $\mathcal{F}$ is restricted to families compatible with these links: the Gaussian distribution for the identity link, and the Bernoulli, binomial, or categorical distribution for the probit and logit links.

We assume the linear predictor for each observation can be decomposed into a time-varying component, $\bm{y}_{3,t}^{*}$, a component $\bm{y}_{2,i}^{*}$ varying between participants, and a component $\bm{y}_{1,it}^{*}$ capturing all remaining variability, yielding
\begin{equation}
    \label{eq:ResponseDecomposition}
    \bm{y}_{it}^{*} = \bm{y}_{1,it}^{*} + \bm{y}_{2,i}^{*} + \bm{y}_{3,t}^{*}.
\end{equation}
The $V_{3}$-dimensional latent variables $\bm{\eta}_{3,t}$ varying between timepoints are related to $\bm{y}_{3,t}^{*}$ through
\begin{equation}
    \begin{aligned}
        \bm{y}_{3,t}^{*}    & = \bm{\nu}_{3} + \bm{\Lambda}_{3} \bm{\eta}_{3,t} + \bm{K}_{3} \bm{X}_{3,t} \\
        \bm{\eta}_{3,t} & = \bm{\alpha}_{3} + \bm{B}_{3}\bm{\eta}_{3,t} + \bm{\Gamma}_{3}\bm{X}_{3,t} + \bm{\xi}_{3,t}, \quad \bm{\xi}_{3,t} \sim \mathcal{N}(\bm{0}, \bm{\Psi}_{3}),
    \end{aligned}
    \label{eq:BetweenTimepointModel}
\end{equation}
where $\bm{\nu}_{3}$ and $\bm{\alpha}_{3}$ are intercepts, $\bm{\Lambda}_{3}$ is a loading matrix, $\bm{B}_{3}$ is a strictly lower triangular matrix for regression between latent variables, $\bm{X}_{3,t}$ is a vector of covariates varying between timepoints, $\bm{K}_{3}$ and $\bm{\Gamma}_{3}$ are matrices with regression coefficients, and $\bm{\Psi}_{3}$ is the covariance matrix of the disturbance term $\bm{\xi}_{3,t}$ for the latent variables. $\mathcal{N}(\bm{a}, \bm{B})$ denotes a normal distribution with mean $\bm{a}$ and covariance matrix $\bm{B}$. 

The $V_{2}$-dimensional latent variables $\bm{\eta}_{2,i}$ varying between participants are related to $\bm{y}_{2,i}^{*}$ through
\begin{equation}
    \begin{aligned}
        \bm{y}_{2,i}^{*}    & = \bm{\nu}_{2} + \bm{\Lambda}_{2} \bm{\eta}_{2,i} + \bm{K}_{2}\bm{X}_{2,i} \\
        \bm{\eta}_{2,i} & = \bm{\alpha}_{2} + \bm{B}_{2} \bm{\eta}_{2,i} + \bm{\Gamma}_{2} \bm{X}_{2,i} + \bm{\xi}_{2,i}, \quad \bm{\xi}_{2,i} \sim \mathcal{N}(\bm{0}, \bm{\Psi}_{2}),
    \end{aligned}
    \label{eq:BetweenIndividualModel}
\end{equation}
where the parameters have similar interpretations as for \eqref{eq:BetweenTimepointModel} except that the covariates $\bm{X}_{2,i}$ and noise $\bm{\xi}_{2,i}$ now vary between participants. Parameters of the within-level model below which vary between participants, such as participant-specific autoregression coefficients, are treated as elements of $\bm{\eta}_{2,i}$, and parameters which vary between timepoints are likewise treated as elements of $\bm{\eta}_{3,t}$.

We define two operators before introducing the within-level model:

\begin{definition}[Lag Operator]
\label{def:lag-operator}
    The lag operator $\mathcal{L}$ is such that $\mathcal{L}^{k} \bm{x}_{t} = \bm{x}_{t-k}$ \citep{kilianStructuralVectorAutoregressive2017}.
\end{definition}

\begin{definition}[Polynomial Matrices]
    \label{def:polynomial-matrices}
    For any collection of matrices $\bm{M}_{0}, \bm{M}_{1}, \dots, \bm{M}_{L}$, let $\bm{M}(\mathcal{L}) = \sum_{l=0}^{L} \bm{M}_{l} \mathcal{L}^{l}$ \citep{kilianStructuralVectorAutoregressive2017}.
\end{definition}

The within-level model relates the $V_{1}$-dimensional latent variables $\bm{\eta}_{1,it}$ varying both between participants and timepoints to $\bm{y}_{1,it}^{*}$,
\begin{equation}
    \begin{aligned}
        \bm{y}_{1,it}^{*}    & = \bm{\nu}_{1,it} + \bm{\Lambda}_{1,it}(\mathcal{L}) \bm{\eta}_{1,it} + \bm{R}_{it}(\mathcal{L}) \bm{y}_{1,it}^{*} + \bm{K}_{1,it} \bm{X}_{1,it} \\
        \bm{\eta}_{1,it} & = \bm{\alpha}_{1,it} + \bm{B}_{1,it}(\mathcal{L}) \bm{\eta}_{1,it} + \bm{Q}_{it}(\mathcal{L}) \bm{y}_{1,it}^{*} + \bm{\Gamma}_{1,it} \bm{X}_{1,it} + \bm{\xi}_{1,it}, \quad \bm{\xi}_{1,it} \sim \mathcal{N}(\bm{0}, \bm{\Psi}_{1,it}),
    \end{aligned}
    \label{eq:WithinLevelModel}
\end{equation}
where we have used Definitions \ref{def:lag-operator} and \ref{def:polynomial-matrices} to write
\begin{align*}
     & \bm{\Lambda}_{1,it}(\mathcal{L}) = \sum_{l=0}^{L} \bm{\Lambda}_{1,lit} \mathcal{L}^{l}, \quad \bm{R}_{it}(\mathcal{L}) = \sum_{l=0}^{L} \bm{R}_{lit} \mathcal{L}^{l}, \\
     &  \bm{B}_{1,it}(\mathcal{L}) = \sum_{l=0}^{L} \bm{B}_{1,lit} \mathcal{L}^{l}, \quad \bm{Q}_{it}(\mathcal{L}) = \sum_{l=0}^{L} \bm{Q}_{lit} \mathcal{L}^{l}.
\end{align*}
Inside the sums, $\bm{\Lambda}_{1,lit}$ is the factor loading matrix relating the lag-$l$ latent variables to the linear predictor, $\bm{R}_{lit}$ and $\bm{B}_{1,lit}$ for $l>0$ are autoregression matrices, $\bm{R}_{0it}$ and $\bm{B}_{1,0it}$ are strictly lower triangular matrices for regression between components of $\bm{y}_{1,it}^{*}$ and $\bm{\eta}_{1,it}$, respectively, and the matrix $\bm{Q}_{lit}$ relates the lag-$l$ linear predictor to the latent variables. We assume that the contemporaneous ($l=0$) part of \eqref{eq:WithinLevelModel} has a unique solution for $\bm{\eta}_{1,it}$ and $\bm{y}_{1,it}^{*}$; Supplementary Material 1 states the precise condition, which holds automatically when $\bm{Q}_{0it} = \bm{0}$, as in all models considered in this paper. Since the within-level model \eqref{eq:WithinLevelModel} is defined in terms of the linear predictor $\bm{y}_{1,it}^{*}$ rather than the observations $\bm{y}_{1,it}$ it differs from the standard DSEM even when $H(\cdot)$ is the identity link. The model is, however, equivalent to the one discussed by \citet[p. 363]{asparouhovDynamicStructuralEquation2018} for discrete responses.

\section{State Space Formulation and Data Augmentation}
\label{sec:StateSpaceFormulation}

Theorem 1 of \citet{sorensenEfficientBayesianEstimation2026} showed that the within-level model of the DSEM for continuous responses could be reformulated as a linear Gaussian state space model (LG-SSM). This allows marginalizing out the within-level latent variables $\bm{\eta}_{1,it}$ analytically using a Kalman filter, alleviating the need to sample them. This approach, however, requires a conditionally Gaussian measurement model. Hence, the observation vector $\bm{y}_{it}$, consisting of discrete responses, cannot be used directly. Instead, we rely on data augmentation to map the discrete observations to continuous, conditionally Gaussian pseudo-observations, denoted $\tilde{\bm{y}}_{it}$, whose within-level components we define as $\tilde{\bm{y}}_{1,it} = \tilde{\bm{y}}_{it} - \bm{y}_{2,i}^{*} - \bm{y}_{3,t}^{*}$. 

Since our within-level model \eqref{eq:WithinLevelModel} contains lagged linear predictors $\bm{y}_{1,it}^{*}$ rather than lagged observations $\bm{y}_{1,it}$, Theorem 1 of \citet{sorensenEfficientBayesianEstimation2026} must be slightly modified. For the discrete-response DSEM presented in the previous section, the within-level model \eqref{eq:WithinLevelModel} is equivalent to the LG-SSM
\begin{equation}
    \label{eq:PseudoStateSpaceModel}
    \begin{aligned}
        \tilde{\bm{y}}_{1,it} &= \bm{Z}_{it} \tilde{\bm{\eta}}_{1,it} + \bm{\epsilon}_{1,it}, \quad \bm{\epsilon}_{1,it} \sim \mathcal{N}(\bm{0}, \bm{\Sigma}_{1,it}) \\
        \tilde{\bm{\eta}}_{1,i,t+1} & = \bm{T}_{it} \tilde{\bm{\eta}}_{1,it} + \bm{c}_{it} + \bm{w}_{it}, \quad \bm{w}_{it} \sim \mathcal{N}\left(\bm{0}, \bm{W}_{it}\right),
    \end{aligned}
\end{equation}
where we have defined the augmented state vector
\begin{equation*}
    \tilde{\bm{\eta}}_{1,it} = 
    \begin{bmatrix}
    \bm{\eta}_{1,it}^{T} & \dots & \bm{\eta}_{1,i,t-L+1}^{T} & (\bm{y}_{1,it}^{*})^{T} & \dots & (\bm{y}_{1,i,t-L+1}^{*})^{T}
    \end{bmatrix}^{T}
\end{equation*}
and the matrix $\bm{Z}_{it} = [\bm{0}_{U \times L V_1} \quad \bm{I}_{U \times U} \quad \bm{0}_{U \times (L-1)U}]$. We refer to Supplementary Material 1 for a proof, and for definitions of the matrices $\bm{T}_{it}$, $\bm{c}_{it}$, and $\bm{W}_{it}$ involved in \eqref{eq:PseudoStateSpaceModel}. The definition of the pseudo-observation $\tilde{\bm{y}}_{it}$ and the covariance matrix $\bm{\Sigma}_{1,it}$ of the observation noise $\bm{\epsilon}_{1,it}$ depend entirely on the latent response distributions, about which we provide the details below.

\subsection{Probit Link}

When $H(\cdot)$ is the standard normal cumulative distribution function, we follow the latent response framework of \citet{albertBayesianAnalysisBinary1993}. For Bernoulli distributed binary responses, the $j$th element of the latent response vector $\tilde{\bm{y}}_{it}$ is drawn from a truncated normal distribution centered on the linear predictor $y_{itj}^{*}$,
\begin{equation}
    \tilde{y}_{itj} \sim 
    \begin{cases} 
        \mathcal{TN}_{(0, \infty)}(y_{itj}^{*}, 1) & \text{if } y_{itj} = 1 \\ 
        \mathcal{TN}_{(-\infty, 0]}(y_{itj}^{*}, 1) & \text{if } y_{itj} = 0,
    \end{cases}
    \label{eq:ProbitGibbs}
\end{equation}
for $j=1,\dots,U$, where $\mathcal{TN}_{(a,b)}(\mu, \sigma^2)$ denotes a normal distribution with mean $\mu$ and variance $\sigma^2$ truncated to the interval $(a,b)$. For ordinal responses, the boundaries $0$ and $\infty$ are replaced by the corresponding adjacent threshold parameters contained in $\bm{\theta}_{it}$. Because the variance of the pseudo-response is fixed to 1 for identification, the observation noise covariance matrix for the LG-SSM \eqref{eq:PseudoStateSpaceModel} is simply the identity matrix, $\bm{\Sigma}_{1,it} = \bm{I}$.

\subsection{Logit Link}

When $H(\cdot)$ is the logistic cumulative distribution function, the model uses a logit link and we employ the Pólya-Gamma data augmentation scheme introduced by \citet{polsonBayesianInferenceLogistic2013}. A random variable $X$ is said to be Pólya-Gamma distributed with parameters $b>0$ and $c \in \mathbb{R}$, $X \sim \mathcal{PG}(b,c)$ if it is equal in distribution to the infinite sum
\begin{equation}
\label{eq:PolyaGammaSum}
    \frac{1}{2 \pi^{2}} \sum_{k=1}^{\infty} \frac{g_{k}}{(k-1/2)^{2} + c^{2}/(4\pi^{2})},
\end{equation}
where $g_{k}$ is gamma distributed with shape $b$ and rate $1$, $g_{k} \sim \mathcal{G}(b, 1)$. \citet{polsonBayesianInferenceLogistic2013} show how $\mathcal{PG}(1,c)$ distributed random variables can be efficiently obtained with a rejection sampler whose probability of acceptance is uniformly bounded below at $0.99919$, and how $\mathcal{PG}(b,c)$ distributed random variables with integer $b$ can be obtained as sums of $b$ such draws. \citet{windleSamplingPolyaGammaRandom2014} proposed alternative samplers which are faster for $b>1$ and also allow non-integer $b$. We used the default hybrid sampler of the \textsc{polyagamma} package \citep{bleki_polyagamma}, which picks a method depending on $b$ and $c$. For Bernoulli responses ($b=1$) it always uses the exact sampler of \citet{polsonBayesianInferenceLogistic2013}, whereas for $b \geq 8$, which occurs in Section~\ref{sec:example4-var}, it uses the saddle-point approximation of \citet{windleSamplingPolyaGammaRandom2014}, so that the draws are approximate in that case.

For binomial data where $y_{itj}$ represents the number of successes in $n_{itj}$ trials---with the Bernoulli distribution corresponding to the special case $n_{itj}=1$---we first define the shifted response $\kappa_{itj} = y_{itj} - n_{itj}/2$. Next, for each observation we introduce the auxiliary random variable $\omega_{itj}$ drawn from a Pólya-Gamma distribution conditioned on the total linear predictor defined in \eqref{eq:ResponseDecomposition},
\begin{equation}
\omega_{itj} \sim \mathcal{PG}(n_{itj}, y_{itj}^{*}).    
\label{eq:LogitGibbs}
\end{equation}
As shown by \citet{polsonBayesianInferenceLogistic2013}, conditioning on $\omega_{itj}$ yields a continuous pseudo-observation defined by the ratio of the shifted response to the auxiliary variable,
\begin{equation}
    \tilde{y}_{itj} = \kappa_{itj}/\omega_{itj}.
    \label{eq:LogitTransformation}
\end{equation}
Using the results of \citet{polsonBayesianInferenceLogistic2013}, we have that the measurement covariance matrix in \eqref{eq:PseudoStateSpaceModel} is $\bm{\Sigma}_{1,it} = \text{diag}( 1/\omega_{it1}, \dots, 1/\omega_{itU} )$. Conditioned on $\bm{\omega}_{it}$, the pseudo-observations perfectly satisfy the linear Gaussian measurement equation \eqref{eq:PseudoStateSpaceModel}, allowing exact marginalization of the within-level states. 

\section{A Hybrid NUTS-Gibbs Sampler with State Space Marginalization}
\label{sec:NUTSGibbs}

The NUTS-Kalman algorithm proposed by \citet{sorensenEfficientBayesianEstimation2026} for analytically marginalizing over within-level latent variables can be utilized also with discrete responses, but it requires some modifications. Let $\bm{\Theta}$ denote the set of all population-level parameters, let $\bm{\eta}_{3} = \{\bm{\eta}_{3,1}, \dots, \bm{\eta}_{3,T}\}$ be the set of all timepoint-specific parameters, $\bm{\eta}_{2}=\{\bm{\eta}_{2,1}, \dots, \bm{\eta}_{2,N}\}$ the set of all participant-specific parameters, $\tilde{\bm{\eta}}_{1,i} = \{\tilde{\bm{\eta}}_{1,i1}, \dots, \tilde{\bm{\eta}}_{1,iT}\}$ the set of all augmented within-level latent variables in the state-space formalization \eqref{eq:PseudoStateSpaceModel} for participant $i$, $\tilde{\bm{y}}_{i} = \{\tilde{\bm{y}}_{i1}, \dots, \tilde{\bm{y}}_{iT}\}$ the set of all latent responses for participant $i$, $\tilde{\bm{y}} = \{\tilde{\bm{y}}_{1}, \dots, \tilde{\bm{y}}_{N}\}$ the set of all latent responses, and similarly $\bm{y} = \{\bm{y}_{1}, \dots, \bm{y}_{T}\}$ and $\bm{y}^{*}= \{\bm{y}_{1}^{*}, \dots, \bm{y}_{T}^{*}\}$. 

At a high level, the proposed algorithm runs as follows. The details of forward filtering backward sampling (FFBS) \citep{carterGibbsSamplingState1994,fruhwirth-schnatterDataAugmentationDynamic1994} and marginalization with the Kalman filter and the NUTS algorithm are described below.
\begin{enumerate}
    \item \textbf{Initialization:} Initialize parameters $\bm{\Theta}^{(0)}$, $\bm{\eta}_{2}^{(0)}$, $\bm{\eta}_{3}^{(0)}$, and $\tilde{\bm{y}}^{(0)}$.
    \item \textbf{Repeat:} Repeat steps 3--4 for $k=1,2,\dots$ until convergence.
    \item \textbf{NUTS step:}
    \begin{enumerate}[label=(\roman*)]
        \item Run one step of NUTS targeting the marginal posterior
        \begin{equation*}
            p\left\{\bm{\Theta}, \bm{\eta}_{2}, \bm{\eta}_{3} | \tilde{\bm{y}}^{(k-1)}\right\} \propto p\left(\bm{\Theta}, \bm{\eta}_{2}, \bm{\eta}_{3}\right) \int p\left\{\tilde{\bm{y}}^{(k-1)} | \bm{\Theta}, \tilde{\bm{\eta}}_{1}, \bm{\eta}_{2}, \bm{\eta}_{3}\right\} p\left(\tilde{\bm{\eta}}_{1} | \bm{\Theta}, \bm{\eta}_{2}, \bm{\eta}_{3}\right) \text{d}\tilde{\bm{\eta}}_{1},
        \end{equation*}
        with the integral in the above expression computed exactly using the Kalman filter. This yields updated values $\{\bm{\Theta}^{(k)}, \bm{\eta}_{2}^{(k)}, \bm{\eta}_{3}^{(k)}\}$.
    \end{enumerate}
    \item \textbf{Gibbs step:}
    \begin{enumerate}[label=(\roman*)]
        \item Using FFBS, sample augmented within-level latent variables conditionally on the parameter values from the NUTS step:
        \begin{equation*}
            \tilde{\bm{\eta}}_{1}^{(k)} \sim p\left\{\tilde{\bm{\eta}}_{1} | \bm{\Theta}^{(k)}, \bm{\eta}_{2}^{(k)}, \bm{\eta}_{3}^{(k)}, \tilde{\bm{y}}^{(k-1)}\right\}.
        \end{equation*}
        \item Compute the linear predictor $\bm{y}_{it}^{*} = \bm{y}_{1,it}^{*} + \bm{y}_{2,i}^{*} + \bm{y}_{3,t}^{*}$ conditionally on $\{\bm{\Theta}^{(k)}, \bm{\eta}_{2}^{(k)}, \bm{\eta}_{3}^{(k)}, \tilde{\bm{\eta}}_{1}^{(k)}\}$ using the definitions in \eqref{eq:BetweenTimepointModel}, \eqref{eq:BetweenIndividualModel}, and \eqref{eq:WithinLevelModel}.
        \item Using \eqref{eq:ProbitGibbs} with a probit link or \eqref{eq:LogitGibbs}--\eqref{eq:LogitTransformation} with a logit link, sample latent responses $\tilde{\bm{y}}_{it}^{(k)} \sim p(\tilde{\bm{y}}_{it} | \bm{y}_{it}^{*})$ for $i=1,\dots,N$ and $t=1,\dots,T$.        
    \end{enumerate}
\end{enumerate}

This algorithm targets the true posterior distribution. That is, when run for a sufficiently large number of iterations $k$, the sample
\begin{equation*}
    \left\{\bm{\Theta}^{(k)}, \bm{\eta}_{3}^{(k)}, \bm{\eta}_{2}^{(k)}, \tilde{\bm{\eta}}_{1}^{(k)}\right\} , \quad k=1,\dots,K,
\end{equation*}
will eventually be arbitrarily close to being a sample from the true posterior distribution. This holds because steps 3 and 4(i) together form a blocked update of $\{\bm{\Theta}, \bm{\eta}_{2}, \bm{\eta}_{3}, \tilde{\bm{\eta}}_{1}\}$ given $\tilde{\bm{y}}^{(k-1)}$, which first updates $\{\bm{\Theta}, \bm{\eta}_{2}, \bm{\eta}_{3}\}$ with $\tilde{\bm{\eta}}_{1}$ integrated out and then draws $\tilde{\bm{\eta}}_{1}$ from its full conditional distribution, while step 4(iii) is a standard Gibbs update of $\tilde{\bm{y}}$; both leave the joint posterior distribution invariant. The algorithm is thus a partially collapsed Gibbs sampler \citep{vandykPartiallyCollapsedGibbs2008}, and the order of the steps matters: because $\tilde{\bm{\eta}}_{1}$ is integrated out in step 3, it must be sampled conditionally on the updated values $\{\bm{\Theta}^{(k)}, \bm{\eta}_{2}^{(k)}, \bm{\eta}_{3}^{(k)}\}$ before $\tilde{\bm{y}}$ is updated. Posterior draws of the within-level latent variables $\bm{\eta}_{1,it}$ can therefore be extracted from $\tilde{\bm{\eta}}_{1}^{(k)}$ in each iteration without any additional sampling, and posterior means and credible intervals for them are computed from these draws as for any other parameter. The Kalman filter operation in step 3(i) and the FFBS in step 4(i) can be parallelized across participants but not across timepoints, due to the temporal autocorrelation in \eqref{eq:PseudoStateSpaceModel}. Sampling of latent responses in step 4(iii), on the other hand, can be parallelized both across participants and timepoints.

\subsection{Kalman Filter and Forward Filtering Backward Sampling}
\label{sec:KalmanFilterFFBS}

We here present the FFBS algorithm \citep{carterGibbsSamplingState1994,fruhwirth-schnatterDataAugmentationDynamic1994} used in the Gibbs step and the Kalman filter used to compute the marginal posterior over $\tilde{\bm{\eta}}_{1,it}$ in the NUTS step. Because the within-level latent variables and pseudo-responses, $\tilde{\bm{\eta}}_{1,it}$ and $\tilde{\bm{y}}_{1,it}$, are conditionally independent across participants given the population and between-level parameters $\{\bm{\Theta}, \bm{\eta}_{2}, \bm{\eta}_{3}\}$, the state-space marginalization and sampling steps can be executed independently across the $N$ participants.

For a given participant $i$, let $\bm{m}_{it|t-1}$ and $\bm{P}_{it|t-1}$ denote the expected value and covariance matrix of the augmented state $\tilde{\bm{\eta}}_{1,it}$ conditioned on the pseudo-observations up to time $t-1$. Let $\bm{m}_{it|t}$ and $\bm{P}_{it|t}$ denote the corresponding filtered moments conditioned on pseudo-observations up to time $t$. Conditionally on the current values of the parameters, the forward filtering pass iterates through $t = 1, \dots, T$ with the prediction equations
\begin{equation}
    \begin{aligned}
        \bm{m}_{it|t-1} &= \bm{T}_{i,t-1} \bm{m}_{i,t-1|t-1} + \bm{c}_{i,t-1}, \\
        \bm{P}_{it|t-1} &= \bm{T}_{i,t-1} \bm{P}_{i,t-1|t-1} \bm{T}_{i,t-1}^{T} + \bm{W}_{i,t-1},
\end{aligned}
\label{eq:KFPrediction}
\end{equation}
where the initial state moments $\bm{m}_{i1|0}$ and $\bm{P}_{i1|0}$ are determined by the implied stationary distribution or set to priors \citep{durbinTimeSeriesAnalysis2012}; the choices made in our examples are given in Sections~\ref{sec:Simulations} and \ref{sec:ApplicationExample}. At each timepoint $t$, we compute the prediction error $\bm{v}_{it}$ and its covariance matrix $\bm{F}_{it}$,
\begin{equation}
    \label{eq:KalmanMeasurement1}
    \begin{aligned}
    \bm{v}_{it} &= \tilde{\bm{y}}_{1,it} - \bm{Z}_{it} \bm{m}_{it|t-1},  \\
    \bm{F}_{it} &= \bm{Z}_{it} \bm{P}_{it|t-1} \bm{Z}_{it}^{T} + \bm{\Sigma}_{1,it}. 
\end{aligned}
\end{equation}
The filtered state moments are subsequently updated via the Kalman gain $\bm{K}_{it} = \bm{P}_{it|t-1} \bm{Z}_{it}^{\top} \bm{F}_{it}^{-1}$, yielding
\begin{equation}
    \label{eq:KalmanMeasurement2}
 \begin{aligned}
    \bm{m}_{it|t} &= \bm{m}_{it|t-1} + \bm{K}_{it} \bm{v}_{it}, \\
    \bm{P}_{it|t} &= \bm{P}_{it|t-1} - \bm{K}_{it} \bm{Z}_{it} \bm{P}_{it|t-1}.
\end{aligned}   
\end{equation}

The forward pass serves two roles in the proposed hybrid algorithm. First, for the NUTS step, it analytically marginalizes out the within-level states to compute the target log-density. The exact log marginal likelihood of the pseudo-observations is given by $\log p(\tilde{\bm{y}} | \bm{\Theta}, \bm{\eta}_{2}, \bm{\eta}_{3}) = \sum_{i=1}^{N} \sum_{t=1}^{T} \log \mathcal{N}(\bm{v}_{it} | \bm{0}, \bm{F}_{it})$. Second, for the Gibbs step, the moments from a forward pass at the parameter values returned by the NUTS step are used to jointly sample the augmented state trajectory via backward sampling. Starting at the final timepoint $T$, a draw is taken from the filtered distribution, $\tilde{\bm{\eta}}_{1,iT} \sim \mathcal{N}(\bm{m}_{iT|T}, \bm{P}_{iT|T})$. For $t = T-1, T-2, \dots, 1$, states are recursively sampled conditionally on the future drawn state, $(\tilde{\bm{\eta}}_{1,it} | \tilde{\bm{\eta}}_{1,i,t+1}) \sim \mathcal{N}(\bm{h}_{it}, \bm{H}_{it})$, where the backward smoothing moments are
\begin{align*}
    \bm{J}_{it} &= \bm{P}_{it|t} \bm{T}_{it}^{T} \bm{P}_{i,t+1|t}^{-1}, \\
    \bm{h}_{it} &= \bm{m}_{it|t} + \bm{J}_{it} \left(\tilde{\bm{\eta}}_{1,i,t+1} - \bm{T}_{it}\bm{m}_{it|t} - \bm{c}_{it}\right), \\
    \bm{H}_{it} &= \bm{P}_{it|t} - \bm{J}_{it} \bm{P}_{i,t+1|t} \bm{J}_{it}^{T}.
\end{align*}
Finally, the required within-level latent variables $\bm{\eta}_{1,it}$ and the linear predictors $\bm{y}_{1,it}^{*}$ are deterministically extracted from the sampled augmented vector $\tilde{\bm{\eta}}_{1,it}$.

Note that efficient numerical implementations may not exactly perform the matrix operations outlined here, as reformulations may lead to considerable speed-up \citep{jungbackerLikelihoodbasedDynamicFactor2015,koopmanFastFilteringSmoothing2000}.

\subsection{The No-U-Turn Sampler}

NUTS was introduced by \citet{hoffmanNoUTurnSamplerAdaptively2014} and solved a long-standing problem in Hamiltonian Monte Carlo (HMC) \citep{Neal2011}. HMC is a variant of the Metropolis-Hastings algorithm which treats all (say, $d$) model parameters as elements of a potential energy vector $\bm{q}$ and introduces $d$ auxiliary momentum variables in the vector $\bm{p}$. By calculating the exact gradient of the log posterior with respect to the potential energy (i.e., the parameters) using automatic differentiation \citep{griewank2008evaluating}, Hamiltonian dynamics in a $2d$-dimensional phase space can be exactly simulated, yielding proposals for all parameters at once with theoretical acceptance probability $1$. In practice, however, the Hamiltonian dynamics has to be approximated using a leapfrog integrator \citep{betancourtConceptualIntroductionHamiltonian2018}, and due to small numerical errors the acceptance rate is in practice slightly below $1$. The step size of the integrator can be tuned during warm-up to achieve a target acceptance ratio, typically $0.80$ or higher. A more important problem is that to ensure efficient exploration of the parameter space, the integrator has to be stopped while the parameters are far away from the current value, before making a U-turn. As its name indicates, NUTS solves exactly this problem, and it does so by recursively constructing a binary tree whose leaves are the parameter values \citep{hoffmanNoUTurnSamplerAdaptively2014}.

NUTS has repeatedly been demonstrated to yield more efficient inference for complex hierarchical models than, e.g., Gibbs sampling or classical Metropolis-Hastings \citep{betancourtConceptualIntroductionHamiltonian2018,betancourtHamiltonianMonteCarlo2015}, and is now the default algorithm in leading probabilistic programming frameworks like \textsc{Stan} \citep{carpenterStanProbabilisticProgramming2017}, \textsc{NumPyro} \citep{phanComposableEffectsFlexible2019}, \textsc{PyMC} \citep{abril-plaPyMCModernComprehensive2023}, and \textsc{Turing.jl} \citep{fjeldeTuringjlGeneralPurposeProbabilistic2025}. While being highly popular in statistics, \textsc{Stan} only allows sampling with NUTS and not hybrid samplers of the kind proposed here. We hence used \textsc{NumPyro}, which is based on the just-in-time (JIT) compiler \textsc{JAX} \citep{frostigCompilingMachineLearning2019} and uses NUTS with an iterative rather than recursive tree-building algorithm \citep[Appendix A]{phanComposableEffectsFlexible2019}.

\section{Simulation Experiments}
\label{sec:Simulations}

We here present simulation experiments comparing the hybrid NUTS-Gibbs sampler to alternative algorithms for estimating binomial-response DSEMs. In particular, for the first simulation example we implemented---in both \textsc{JAGS} \citep{plummer2003jags} and \textsc{NumPyro}---a Gibbs sampler close to the one proposed by \citet{asparouhovDynamicStructuralEquation2018}, and---in \textsc{NumPyro} and \textsc{Stan}---a pure NUTS algorithm sampling over all parameters with no marginalization. The former algorithm utilized the discrete response formulation, whereas the latter used the logit/probit link directly since in NUTS this would just introduce additional parameters with no benefits.

Simulations were run in \textsc{Python} 3.14.3 on a MacBook Pro with an Apple M1 Max 32GB chip, using \textsc{CmdStanPy} 1.3.0 \citep{cmdstanpy}, \textsc{Stan} 2.36.0, \textsc{JAX} 0.7.2, \textsc{NumPyro} 0.20.0, and \textsc{polyagamma} 2.0.2 \citep{bleki_polyagamma}. An exception was the \textsc{JAGS} 4.3.2 implementation, which was run with \textsc{R} 4.5.2 \citep{rcoreteamLanguageEnvironmentStatistical2025} through \textsc{jagsUI} 1.5.2 \citep{kellnerJagsUIWrapperRjags2026}. Posterior analyses and diagnostics were done in \textsc{ArviZ} 0.23.4 \citep{martinArviZModularFlexible2026}. All simulation results can be reproduced using code openly available in our OSF repository.\footnote{Anonymized view-only link for peer review: \url{https://osf.io/ds3hp/overview?view_only=9e21f9b580914bc9bba7522a17eb7b80}.}

Since all algorithms target the true posterior distribution, they will eventually give accurate output if given enough time and computer memory. The question of relevance for a user is how long one has to wait to get sufficiently accurate estimates of the posterior quantities of interest. To this end, we compared the algorithms in terms of their sampling efficiency, as measured by effective sample size (ESS) per time unit, as well as their potential scale reduction factor $\hat{R}$. We computed the minimum ESS and maximum $\hat{R}$ over the parameters of interest, using the definitions of \citet{vehtariRankNormalizationFoldingLocalization2021}. In particular, we distinguish between bulk-ESS for estimating quantities in the bulk of the posterior (e.g., posterior means) and tail-ESS for estimating quantities in the tails of the posterior (e.g., posterior intervals). 

For each example, we inspected convergence diagnostics and compared posterior means for all parameters across algorithms. Note that with the probit link in the experiment with participant-invariant dynamics, the Gibbs samplers used different priors than the other algorithms and hence targeted a slightly different posterior distribution. All algorithms were run with four parallel chains. The length of the warm-up phases for each algorithm and condition were chosen through initial pilot runs, and the total number of iterations was set long enough to ensure reliable estimation of the effective sampling rate and a sufficiently high bulk-ESS for reliable posterior summaries. Further details regarding sampler diagnostics, as well as additional tables and plots, can be found in Supplementary Material 2.

\subsection{Bernoulli Response Five-Indicator AR(1) Models}
\label{sec:five-ind-ar1}

We first simulated data from a two-level DSEM with $U = 5$ binary indicators, a between-participant latent trait and a within-participant latent trait. There was no systematic between-timepoint variation, so $\bm{y}_{3,t}^{*}=\bm{0}$. The between-participants model was
\begin{equation}
    \label{eq:example1-between-model}
        \bm{y}_{2,i}^{*} = \bm{\nu} + \bm{\lambda}_{2} \eta_{2,i}, \quad 
        \eta_{2,i}  \sim \mathcal{N}(0, \psi_{2}^{2}), 
\end{equation}
where the variance of the between-level latent trait was $\psi_{2}^{2}=0.5$ and the intercepts were $\bm{\nu} = (-1, -0.5, 0, 0.5, 1)^{T}$. The within-level model was a first-order autoregressive (AR(1)) process for the latent trait combined with a reflective measurement model for the linear predictor,
\begin{equation}
\label{eq:example1-within-model}
    \begin{aligned}
    \bm{y}_{1,it}^{*} &= \bm{\lambda}_{1,i} \eta_{1,it} \\
    \eta_{1,it} &= \phi_{1,i} \eta_{1,i,t-1} + \xi_{1,it}, \quad \xi_{1,it} \sim \mathcal{N}(0, \psi_{1,i}^{2}).
\end{aligned}
\end{equation}
In \eqref{eq:example1-between-model} and \eqref{eq:example1-within-model}, $\bm{\lambda}_{1,i}$ and $\bm{\lambda}_{2}$ are factor loading vectors, whose first elements were fixed to unity for identifiability, $\lambda_{1,i1}=\lambda_{2,1}=1$, and $\phi_{1,i}$ is the autoregression coefficient. In all estimation algorithms, the initial within-level state was given its stationary distribution, $\eta_{1,i1} \sim \mathcal{N}\{0, \psi_{1,i}^{2}/(1-\phi_{1,i}^{2})\}$. 

We generated data randomly from two different variants of \eqref{eq:example1-between-model}--\eqref{eq:example1-within-model} and compared the efficiency of the proposed hybrid NUTS-Gibbs sampler to alternative algorithms, as described next.

\subsubsection{Participant-Invariant Dynamics}

First, we let the dynamics be participant-invariant, by fixing $\phi=0.4$ and $\psi_{1}^{2} = 1-\phi^{2} = 0.84$, ensuring that the marginal variance of the within-level state is exactly $\psi_{1}^{2} / (1-\phi^{2}) = 1$ \citep[Ch. 5.6.4]{durbinTimeSeriesAnalysis2012}. We also fixed the factor loadings to common values by sampling $\lambda_{1,j}$ and $\lambda_{2,j}$ ($j=2,\dots,5$) uniformly between 0.6 and 1.2 for each generated dataset. In this section, we omit $i$ subscripts for all parameters which did not vary between participants. The responses were Bernoulli distributed with either a probit link or a logit link.

We analyzed the data with a model that also treated the dynamics as participant-invariant. This way, it was possible to define a pure Gibbs sampler with no Metropolis step, and hence no need to manually adjust proposal distributions. The Gibbs sampler was implemented in \textsc{JAGS} and \textsc{NumPyro}, for the probit link only, using the truncated normal latent response formulation \eqref{eq:ProbitGibbs}. We limited the Gibbs sampler implementations to the probit link because, while it is in principle straightforward to implement Pólya-Gamma distributed latent responses in the Gibbs sampler, we are not aware of any publications in the literature discussing this, and the purpose of this paper is not to develop new pure Gibbs samplers. To ensure conjugacy we used the priors $\phi \sim \mathcal{TN}_{[-0.999, 0.999]}(0, 1)$, $\psi_{1}^{2} \sim \mathcal{IG}(0.01, 0.01)$, and $\psi_{2}^{2} \sim \mathcal{IG}(0.01, 0.01)$, where $\mathcal{IG}(a, b)$ denotes an inverse gamma distribution with shape $a$ and scale $b$. We also implemented a pure NUTS algorithm in both \textsc{NumPyro} and \textsc{Stan}, which used the observed responses directly and sampled over all latent states. Both NUTS and the proposed hybrid NUTS-Gibbs sampler used the prior $\tanh^{-1}(\phi) \sim \mathcal{N}(0, 1)$ to ensure stationarity of the AR(1) process \eqref{eq:example1-within-model}, and log-normal priors $\log \psi_{1}^{2} \sim \mathcal{N}(0, 1)$ and $\log \psi_{2}^{2} \sim \mathcal{N}(0, 1)$ to ensure positive variances. For the remaining parameters all algorithms used the priors $\bm{\nu} \sim \mathcal{N}(\bm{0}, 4 \cdot \bm{I})$, and $\lambda_{1,j} \sim \mathcal{N}(1, 0.25)$ and $\lambda_{2,j} \sim \mathcal{N}(1, 0.25 )$ for $j=2,\dots,5$. Due to the different priors the Gibbs samplers targeted a slightly different posterior than the other algorithms, but we confirmed that all implementations gave posterior means close to the data-generating values. 

In the first experiment, we kept the number of participants fixed at $N=20$ while gradually increasing the number of timepoints to $T=50,200,500,2500$. Here and in the remaining simulations, five datasets were generated for each condition. Due to long run times, the \textsc{JAGS} and \textsc{Stan} implementations were only run up to $T=200$, while only the hybrid sampler was run with $T=2500$. For the remaining simulations, we needed to impose similar limitations due to long run times, which the plots clearly highlight below.

Figure \ref{fig:example1-efficiency} (left) compares the bulk efficiency of the algorithms with the probit link. It shows that the hybrid sampler dominates the other algorithms for any number of timepoints. Its efficiency gain compared to the second best algorithm (the \textsc{NumPyro} NUTS implementation) was $1.7$ with $T=50$, $3.2$ with $T=200$, $2.4$ with $T=500$, and $2.3$ with $T=2500$. With the logit link, the performance gain of the hybrid sampler was even stronger: in this case, it was $6.9$ times more efficient than NUTS with $T=50$, $3.2$ times with $T=200$, and $2.6$ times with $T=500$. This means a model with $N=20$ and $T=500$ runs in $10.8$ instead of $26.0$ minutes with a probit link and $7.1$ instead of $18.5$ minutes with a logit link.

In the second experiment, we kept the number of timepoints fixed at $T=50$ and increased the number of participants to $N = 50, 100, 200, 500$. Again, the proposed hybrid sampler dominated all other algorithms in terms of efficiency, as shown in Figure \ref{fig:example1-efficiency} (right). In this case, however, the Gibbs sampler implemented in \textsc{NumPyro} was second best in all cases. The efficiency gain of the hybrid sampler was $1.9$ for all $N$. With the logit link, for which no Gibbs samplers were implemented, the \textsc{NumPyro} NUTS implementation was always the second best algorithm. The hybrid sampler was $10.2$ times more efficient with $N=50$, $14.8$ times more efficient with $N=100$, and $16.6$ times more efficient with $N=200$. This means a model with $N=200$ and $T=50$ runs in $2.6$ instead of $4.8$ minutes with a probit link and $52$ seconds instead of $14.4$ minutes with a logit link.

It is also interesting to note from Figure \ref{fig:example1-efficiency} that \textsc{NumPyro} is consistently more efficient than \textsc{Stan}, similar to what has been observed by others \citep{phanComposableEffectsFlexible2019}. The difference in efficiency between the \textsc{NumPyro} and \textsc{JAGS} implementations of the Gibbs sampler is even more pronounced.

\begin{figure}
    \includegraphics[width=\textwidth]{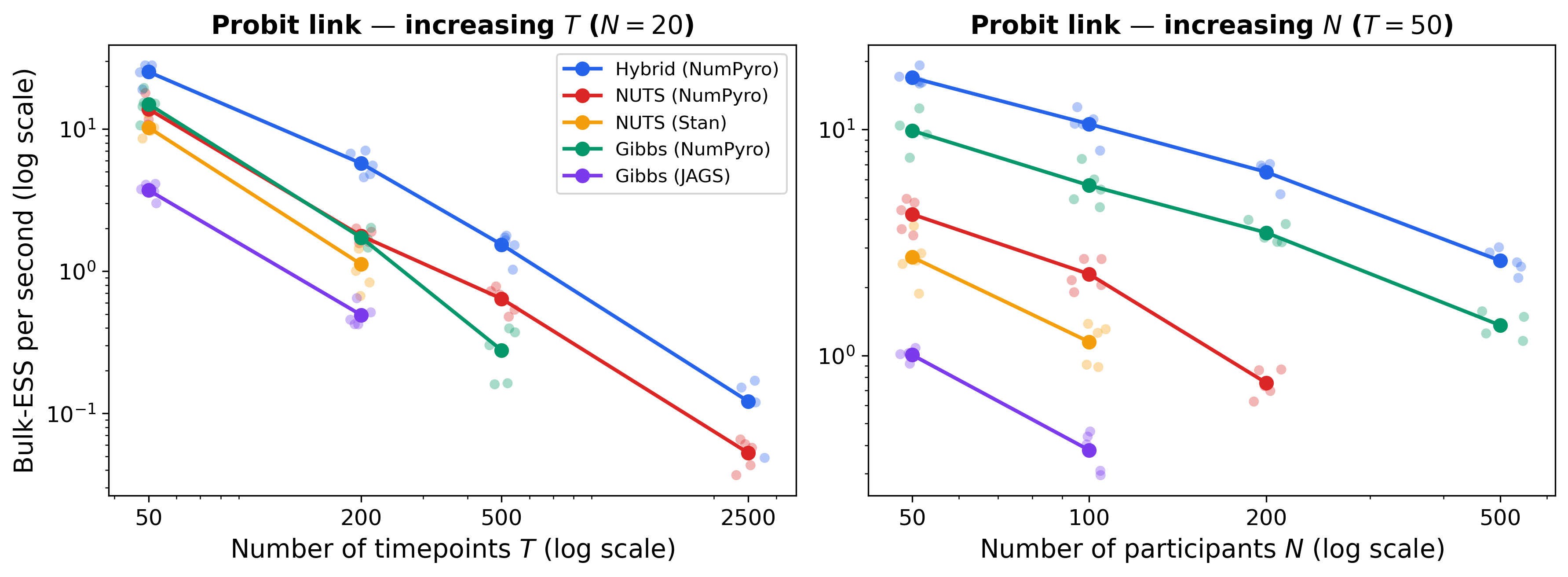}
    \caption{Bulk efficiency for the five-indicator AR(1) model with participant-invariant dynamics. Means across five runs are shown, together with results from each individual run as semi-transparent points}
    \label{fig:example1-efficiency}    
\end{figure}

\subsubsection{Participant-Varying Dynamics}
\label{sec:example2-bernoulli-participant-varying}

We next allowed the autoregressive coefficient, the within-level variance, and the factor loadings for the within-level latent traits, to vary between participants, following the priors
\begin{equation*}
        \tanh^{-1}{(\phi_{i})} \sim \mathcal{N}(\mu_{\phi}, \omega_{\phi}^2), \quad 
    \log(\psi_{1,i}^{2}) \sim \mathcal{N}(\mu_{\psi}, \omega_{\psi}^2), \quad \text{and} \quad \lambda_{1,ij} \sim \mathcal{N}(\mu_{\lambda,j}, \omega_{\lambda}^2), ~ j=2,\dots,5.
\end{equation*}
In the data generation, we sampled parameter values from these priors while setting the population-level means equal to the values used for the model with participant-invariant dynamics, i.e., $\mu_{\phi}=\tanh^{-1}(0.4)$, $\mu_{\psi} = \log( 0.84)$, and $\mu_{\lambda,j} \sim \mathcal{U}(-0.6, 1.2)$ for $j=2,\dots,5$ with $\mathcal{U}(a,b)$ denoting the uniform distribution on $[a, b]$. Note that in the general DSEM formulation of equations \eqref{eq:BetweenTimepointModel}--\eqref{eq:WithinLevelModel}, all parameters varying between participants are part of the vector $\bm{\eta}_{2,i}$. Hence, the variable $\eta_{2,i}$ in \eqref{eq:example1-between-model} is actually an element of a larger vector, but for ease of presentation we do not write this out fully.

Since this model does not allow pure Gibbs sampling, we implemented the hybrid NUTS-Gibbs sampler in \textsc{NumPyro} and pure NUTS in \textsc{NumPyro} and \textsc{Stan}. For $\mu_{\phi}$, $\mu_{\psi}$, $\mu_{\lambda,j}$, and $\lambda_{2,j}$ ($j=2,\dots,5$) we used standard normal priors, and for $\omega_{\phi}^{2}$, $\omega_{\psi}^{2}$, and $\omega_{\lambda}^{2}$ we used standard normal priors truncated to the positive real line. All implementations used a non-centered parametrization \citep{papaspiliopoulosGeneralFrameworkParametrization2007}.

We used the same set-up as we did for the model with participant-invariant dynamics. First, we fixed $N=20$ and set the number of timepoints to $T = 50, 200, 500$. Figure \ref{fig:example2-efficiency} (left) shows the results with the logit link. Here the pure NUTS sampler was more efficient than the hybrid sampler, having efficiency gain $1.11$ with $T=50$, $1.25$ with $T=200$, and $2.0$ with $T=500$. In contrast, when we fixed $T=50$ and increased $N$, the hybrid sampler was more efficient, with an efficiency gain between $1.5$ and $1.7$. The results were similar for tail efficiency and with the probit link, as shown in Supplementary Material 2.

\begin{figure}
    \centering
    \includegraphics[width=\linewidth]{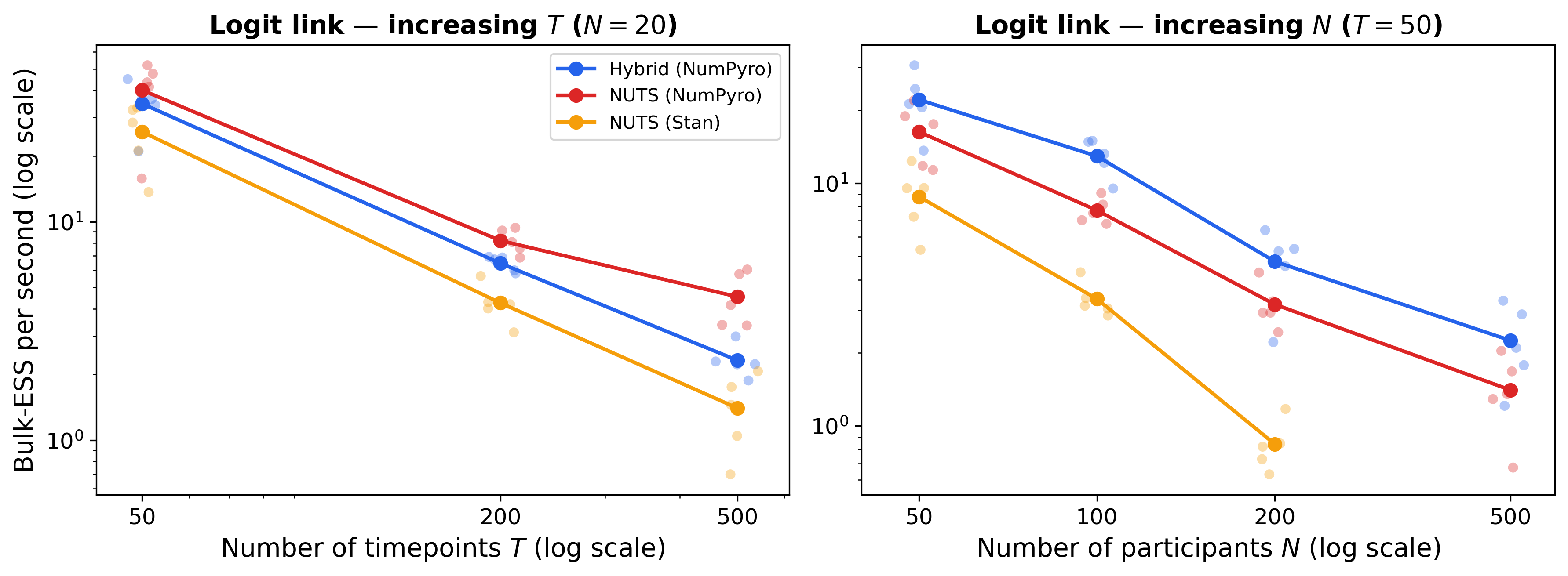}
    \caption{Bulk efficiency for the five-indicator AR(1) model with participant-varying dynamics}
    \label{fig:example2-efficiency}
\end{figure}

\subsection{Multivariate Binomial Response Nine-Indicator VAR(1) Model}
\label{sec:example4-var}

We next considered a first-order vector autoregressive (VAR(1)) model with three-dimensional latent states $\bm{\eta}_{1,it}$ and $\bm{\eta}_{2,i}$ and nine indicators. Again there was no systematic variation between timepoints, so $\bm{y}_{3,t}^{*} = \bm{0}$. The between-participants model was
\begin{align*}
    \bm{y}_{2,i}^{*} = \bm{\nu} + \bm{\Lambda}_{2} \bm{\eta}_{2,i}, \quad \bm{\eta}_{2,i}  \sim \mathcal{N}(\bm{0}, \bm{\Psi}_{2,i}),
\end{align*}
and the within-level model was
\begin{align*}
    \bm{y}_{1,it}^{*} &= \bm{\Lambda}_{1,i} \bm{\eta}_{1,it}  \\
    \bm{\eta}_{1,it} &= \bm{\Phi}_{1,i} \bm{\eta}_{1,i,t-1} + \bm{\xi}_{1,it}, \quad \bm{\xi}_{1,it} \sim \mathcal{N}(\bm{0}, \bm{\Psi}_{1,i}).
\end{align*}
The factor loading matrices were
\begin{equation*}
    \bm{\Lambda}_{1,i} = 
    \begin{bmatrix}
    1 & 0 & 0 \\
    \lambda_{1,i1} & 0 & 0 \\
    \lambda_{1,i2} & 0 & 0 \\
    0 & 1 & 0 \\
    0 & \lambda_{1,i3} & 0 \\
    0 & \lambda_{1,i4} & 0 \\
    0 & 0 & 1 \\
    0 & 0 & \lambda_{1,i5} \\
    0 & 0 & \lambda_{1,i6}
    \end{bmatrix}
    \quad \text{and} \quad 
    \bm{\Lambda}_{2} = 
    \begin{bmatrix}
    1 & 0 & 0 \\
    \lambda_{2,1} & 0 & 0 \\
    \lambda_{2,2} & 0 & 0 \\
    0 & 1 & 0 \\
    0 & \lambda_{2,3} & 0 \\
    0 & \lambda_{2,4} & 0 \\
    0 & 0 & 1 \\
    0 & 0 & \lambda_{2,5} \\
    0 & 0 & \lambda_{2,6}
    \end{bmatrix},
\end{equation*}
with $\bm{\Lambda}_{1,i}$ varying between participants but $\bm{\Lambda}_{2}$ being constant in the population. The elements of $\bm{\Lambda}_{1,i}$ and $\bm{\Lambda}_{2}$ were generated in exactly the same way as for the AR(1) model in Section \ref{sec:example2-bernoulli-participant-varying}. The between-level covariance matrix and the population mean of the autoregressive matrix took the values
\begin{equation*}
    \bm{\Psi}_{2} = 
    0.7 \cdot \begin{bmatrix}
        1 & 0.3 & 0.2 \\
        0.3 & 1 & 0.3 \\
        0.2 & 0.3 & 1
    \end{bmatrix}
    \quad 
    \text{and}
    \quad
    \bm{\mu}_{\bm{\Phi}} = 
    \begin{bmatrix}
        0.3 & 0.1 & 0.1 \\
        0.1 & 0.3 & 0.1 \\
        0.1 & 0.1 & 0.3
    \end{bmatrix}.
\end{equation*}
The individual autocorrelation matrices $\bm{\Phi}_{1,i}$ were then generated by adding zero-mean normally distributed deviations with variance $0.1$ independently around each element of $\bm{\mu}_{\bm{\Phi}}$. To avoid accidentally generating non-stationary processes, for each $\bm{\Phi}_{1,i}$ we checked if the largest absolute eigenvalue $\rho_{i}$ was larger than $0.95$.\footnote{Processes with $\rho_{i} \geq 1$ are non-stationary \citep[Ch. 2.2]{kilianStructuralVectorAutoregressive2017}.} If it was, the scaling $\bm{\Phi}_{1,i} \gets (0.95 / \rho_{i}) \bm{\Phi}_{1,i}$ was applied to strictly constrain the maximum absolute eigenvalue to $0.95$.

The individual process noise covariance matrices $\bm{\Psi}_{1,i}$ were generated by first defining the $3 \times 3$ Cholesky factors $\bm{L}_{1,i}$, whose diagonal elements were sampled log-normally around a location parameter corresponding to a variance of $0.5$,
\begin{equation*}
    \text{diag}(\bm{L}_{1,i}) \sim \mathcal{N}\left\{\log \left[\sqrt{0.5}, \sqrt{0.5}, \sqrt{0.5}\right]^{T}, 0.1 \cdot \bm{I}\right\}.
\end{equation*}
The off-diagonal elements of $\bm{L}_{1,i}$ were sampled from a normal distribution centered at zero with variance $0.1$. The final matrix was then constructed as $\bm{\Psi}_{1,i} = \bm{L}_{1,i} \bm{L}_{1,i}^{T}$. In all estimation algorithms, the initial within-level state was given the distribution $\bm{\eta}_{1,i1} \sim \mathcal{N}(\bm{0}, \bm{P}_{i})$ with $\bm{P}_{i} = \sum_{k=0}^{5} \bm{\Phi}_{1,i}^{k} \bm{\Psi}_{1,i} (\bm{\Phi}_{1,i}^{k})^{T} + 10^{-6} \bm{I}$, a truncated series approximation to the stationary covariance matrix, whereas the data were generated using the exact stationary covariance matrix.

The observed counts were generated from a binomial distribution using a logit link, such that $y_{itj} \sim \mathcal{B}\{n_j, \text{logit}^{-1}(y_{itj}^{*})\}$, with the number of trials $n_{j}$ being elements of the vector $\bm{n} = (3,3,3,1,1,1,9,9,9)^{T}$. We did not include a probit link here, since data augmentation then would require one latent response per trial $n_{j}$, which would drastically increase the memory requirements and the dimensionality of the parameter space.

We implemented the hybrid NUTS-Gibbs sampler and a pure NUTS sampler in \textsc{NumPyro}. A NUTS algorithm implemented in \textsc{Stan} was also tested, but discarded due to excessive run times. Again, no pure Gibbs sampler can be defined due to lack of conjugacy. The prior distributions were multivariate extensions of those used in the previous section, and are listed in Supplementary Material 2.

This is a challenging, high-dimensional model: in addition to the large number of parameters varying between participants which are common to both algorithms, the hybrid sampler needs to sample $18 NT$ parameters $\{\omega_{itj}, \tilde{y}_{itj}\}$ in each iteration while NUTS has to sample $3 N T$ latent variables $\eta_{1,itj}$. Note the important difference, however, that while the hybrid sampler obtains these parameters from conditional distributions, NUTS has to repeatedly compute the derivative of the log posterior with respect to $\mathcal{O}(N \cdot T)$ parameters in order to simulate Hamiltonian dynamics with the leapfrog integrator. In contrast, the hybrid sampler only simulates Hamiltonian dynamics for $\mathcal{O}(N+T)$ parameters.

Based on several pilot runs, we found that data with $N=20$ participants and $T=50$ timepoints were feasible to analyze using four parallel chains on the same MacBook Pro as used in the preceding experiments. We ran 4,000 iterations of each chain, discarding the first 1,000 as warm-up. Table \ref{tab:example4-summary} shows the average results over five generated datasets. The hybrid sampler is $5.1$ times more efficient in terms of bulk-efficiency and $4.7$ times more efficient in terms of tail efficiency. Assuming the algorithms would have similar efficiency if they kept running, it would take about $2.8$ hours to obtain a bulk-ESS of 1,000 with NUTS, compared to 33 minutes with the hybrid sampler.

\begin{table}[htpb]
    \centering
    \caption{Benchmark results for the nine-indicator VAR(1) model with binomial responses. Averages over 5 runs}
    \label{tab:example4-summary}
    \begin{tabular}{lrrrrrr}
        \toprule
        \textbf{Algorithm} & \textbf{Wall time} & \textbf{Bulk-ESS} & \textbf{Tail-ESS} & \textbf{Bulk-ESS/s} & \textbf{Tail-ESS/s} & $\hat{\boldsymbol{R}}$ \\
        \midrule
        \textbf{Hybrid} & 30 min & 846 & 1355 & 0.51 & 0.84 & 1.006 \\
        \textbf{NUTS} & 25 min & 153 & 260 & 0.10 & 0.18 & 1.036 \\
        \bottomrule
    \end{tabular}
\end{table}

\section{Application Example: Predicting Panic Attacks from Physical Activity}
\label{sec:ApplicationExample}

We now present an application example using a dataset from \citet{jangDigitalPhenotypingDataset2024}, downloaded from the OpenESM database \citep{siepeIntroducingOpenESMDatabase2025}. The study followed 43 patients diagnosed with mental health disorders over a period of 402 days (i.e., a low $N$, high $T$ data set). The patients were asked to complete a questionnaire on a smartphone app once per day, in the evening. The average number of observations per participant was 92, varying from a minimum of 1 to a maximum of 327. 

One of the collected variables was a binary indicator reporting whether the participant had experienced a panic attack on the given day. We denote this variable by $y_{it1}$. We are interested in investigating the extent to which such a panic attack can be predicted by daily mean heart rate ($y_{it2}$), heart rate variance ($y_{it3}$), and daily step count ($y_{it4}$). The missingness pattern in the data is shown in Figure \ref{fig:missingness-pattern}. Two participants had only a single observation, and were removed from the analyses.

\begin{figure}
    \centering
    \includegraphics[width=\linewidth]{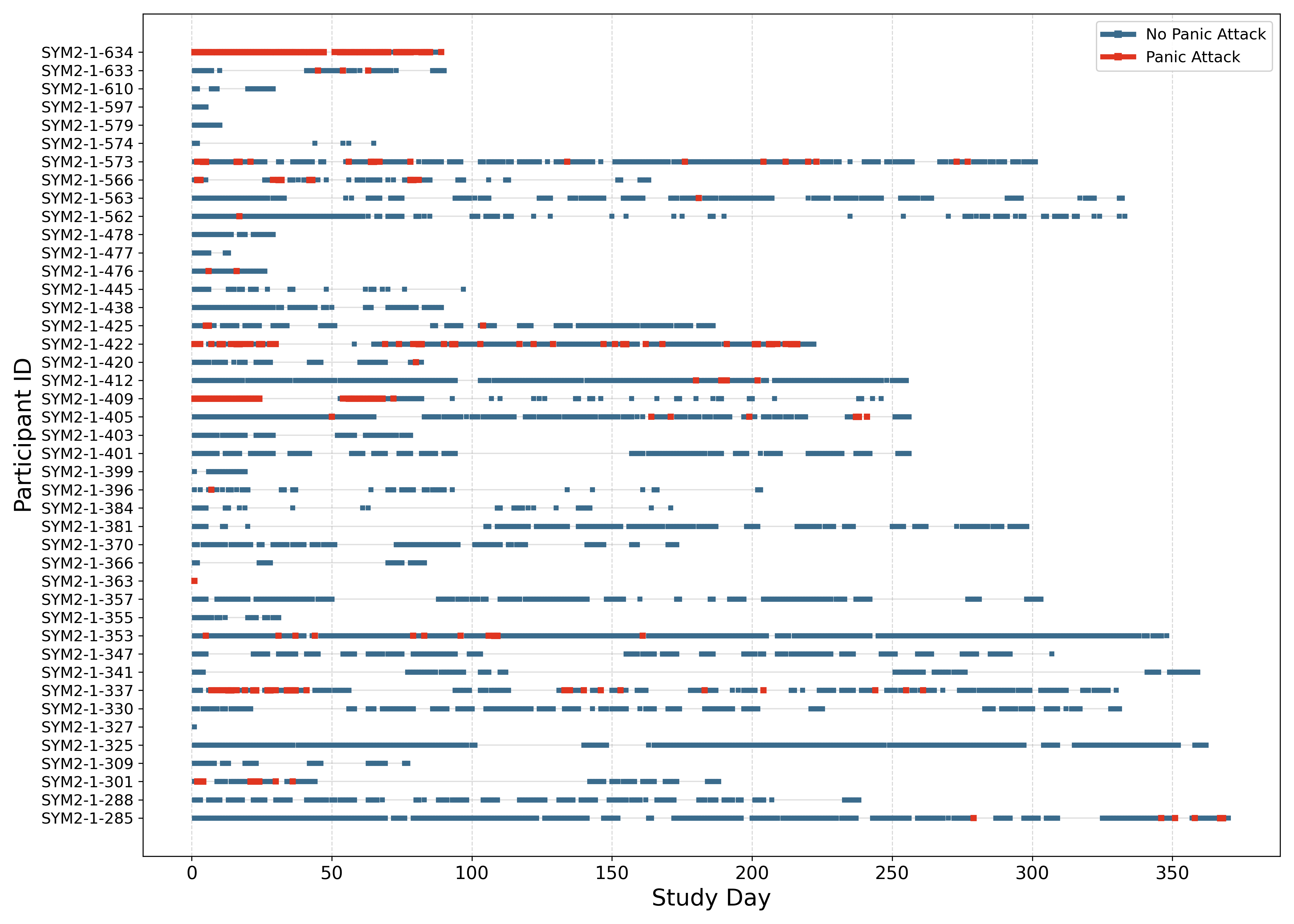}
    \caption{Missingness pattern in the EMA data. The presence of panic attacks is indicated in red}
    \label{fig:missingness-pattern}
\end{figure}

No systematic variation between timepoints was assumed, so $\bm{y}_{3,t}^{*} = \bm{0}$. The between-participants model was
\begin{align}
\label{eq:ApplicationBetween}
    \bm{y}_{2,i}^{*} = \bm{\nu} + \bm{\eta}_{2,i}, \quad \bm{\eta}_{2,i} \sim \mathcal{N}(\bm{0}, \bm{\Psi}_{2}),
\end{align}
where $\bm{\nu}$ is an intercept vector and $\bm{\Psi}_{2}$ a diagonal covariance matrix. The within-participants model was
\begin{equation*}
    \begin{aligned}
    \bm{y}_{1,it}^{*} &= \bm{\eta}_{1,it} \\    
    \bm{\eta}_{1,it}
    &=
    \bm{\Phi}_{1}
    \bm{\eta}_{1,i,t-1}
    +
    \bm{\xi}_{1,it},
    \quad \bm{\xi}_{1,it} \sim \mathcal{N}(\bm{0}, \bm{\Psi}_{1}).
\end{aligned}
\end{equation*}

The panic attack measurements were modeled with a Bernoulli distribution with mean 
\begin{equation*}
    \mu_{it1} = H(y_{it1}^{*}) = H(y_{1,it1}^{*} + y_{2,i1}^{*}) = H(\nu_{1} + \eta_{2,i1} + \eta_{1,it1}),
\end{equation*}
where $H(\cdot)$ is the inverse logit function. We assumed the continuous measurements were conditionally normally distributed, so $y_{itj} = y_{itj}^{*} + \epsilon_{itj}$ with $\epsilon_{itj} \sim \mathcal{N}(0, \sigma_{j}^{2})$ for $j=2,3,4$. The process covariance matrix $\bm{\Psi}_{1}$ was diagonal and the autoregression matrix was
\begin{equation*}
\bm{\Phi}_{1} = 
    \begin{bmatrix}
        \phi_{11} & \phi_{12} & \phi_{13} & \phi_{14} \\
        0 & \phi_{22} & 0 & 0 \\
        0 & 0 & \phi_{33} & 0 \\
        0 & 0 & 0 & \phi_{44}
    \end{bmatrix},
\end{equation*}
indicating that we assume all variables are auto-correlated, and that there are cross-lagged effects of the continuous variables on the probability of panic attacks. The initial within-level state was given the distribution $\bm{\eta}_{1,i1} \sim \mathcal{N}(\bm{0}, \bm{\Psi}_{1} + 10^{-4} \bm{I})$.

The high number of missing values were conveniently handled with the Kalman filter as follows. Whenever an observation $\bm{y}_{it}$ was missing---note that this includes the binary panic attack variable and the three continuous measurements---we updated the corresponding within-level latent variables $\bm{\eta}_{1,it}$ with the prediction step \eqref{eq:KFPrediction} and set $\bm{y}_{1,it}^{*} = \bm{\eta}_{1,it}$, skipping the measurement update steps \eqref{eq:KalmanMeasurement1} and \eqref{eq:KalmanMeasurement2}. This treats the missing values as missing at random \citep{rubinInferenceMissingData1976}, that is, it assumes that whether a value is missing may depend on observed values but not on the missing values themselves. This assumption may be violated in the present data, e.g., if participants were less likely to respond on days with panic attacks, in which case the estimates could be biased. 

We ran the hybrid NUTS-Gibbs sampler for 8,000 iterations on each of four parallel chains. This took slightly less than 2 hours. The posterior summary is shown on Table \ref{tab:application-summary}, in which $\psi_{1,j}$ and $\psi_{2,j}$ denote the square roots of the diagonal elements of $\bm{\Psi}_{1}$ and $\bm{\Psi}_{2}$, respectively.

While most $\hat{R}$ values show excellent convergence, we note that the sampler struggles with identifying the intercept $\nu_{1}$ and standard deviation $\psi_{21}$ of the between-participant part of the linear predictor for panic attacks defined in \eqref{eq:ApplicationBetween}, as both of these have $\hat{R}=1.02$ and bulk-ESS 291 and 328, respectively. The same applies to the within-level process noise of the linear predictor for panic attacks, $\psi_{1,1}$, which has $\hat{R}=1.01$ and bulk-ESS 258. Note, however, that these $\hat{R}$ values are much better than the $\hat{R} \leq 1.10$ thresholds which have traditionally been suggested in the literature \citep[p. 297]{Gelmanetal2003} and used in the DSEM literature \citep{asparouhov2010bayesian,asparouhovDynamicStructuralEquation2018,mcneishExploringHowMany2025}. Considering Figure \ref{fig:missingness-pattern}---which demonstrates that panic attacks are rare and unevenly distributed---it is clear that the data contain relatively sparse information about panic attacks, which explains why these parameters are hard to identify. As expected, the overall intercept for panic attacks is very low, with posterior mean $-6.123$ and 95\% posterior intervals $[-7.836, -4.650]$. We also see that the autoregressive effects $\phi_{jj}$ ($j=1,\dots,4$) are all very high, indicating a high degree of similarity in these variables from one day to the next. Looking at the parameters of main interest, $\phi_{12}$, $\phi_{13}$, and $\phi_{14}$, we see that all contain $0$ in their posterior interval. Hence, we cannot conclude that panic attacks can be predicted by mean heart rate, heart rate variance, and step count the day before.

\begin{table}[htpb]
    \centering
    \caption{Posterior summary for the application example. The 2.5\% and 97.5\% columns show the limits of the 95\% credible interval}
    \label{tab:application-summary}
    \begin{tabular}{lrrrrrrr}
        \toprule
        \textbf{Parameter} & \textbf{Mean} & \textbf{SD} & \textbf{2.5\%} & \textbf{97.5\%} & \textbf{Bulk-ESS} & \textbf{Tail-ESS} & $\hat{\boldsymbol{R}}$ \\
        \midrule
        $\phi_{11}$ & 0.922 & 0.038 & 0.838 & 0.984 & 403 & 607 & 1.00 \\
        $\phi_{22}$ & 0.871 & 0.024 & 0.822 & 0.916 & 9444 & 13869 & 1.00 \\
        $\phi_{33}$ & 0.967 & 0.019 & 0.919 & 0.992 & 8492 & 14285 & 1.00 \\
        $\phi_{44}$ & 0.921 & 0.015 & 0.890 & 0.947 & 14530 & 17932 & 1.00 \\
        $\phi_{12}$ & 0.064 & 0.077 & -0.076 & 0.230 & 1385 & 2586 & 1.00 \\
        $\phi_{13}$ & 0.198 & 0.250 & -0.205 & 0.787 & 923 & 2044 & 1.00 \\
        $\phi_{14}$ & 0.029 & 0.074 & -0.111 & 0.185 & 2397 & 4672 & 1.00 \\
        $\nu_{1}$ & -6.123 & 0.816 & -7.836 & -4.650 & 291 & 706 & 1.02 \\
        $\nu_{2}$ & 0.033 & 0.137 & -0.235 & 0.302 & 2449 & 5724 & 1.00 \\
        $\nu_{3}$ & -0.006 & 0.099 & -0.197 & 0.190 & 3981 & 8409 & 1.00 \\
        $\nu_{4}$ & 0.036 & 0.092 & -0.143 & 0.218 & 5523 & 11259 & 1.00 \\
        $\sigma_{2}$ & 0.473 & 0.012 & 0.448 & 0.497 & 13874 & 20282 & 1.00 \\
        $\sigma_{3}$ & 0.767 & 0.010 & 0.748 & 0.785 & 31012 & 21393 & 1.00 \\
        $\sigma_{4}$ & 0.660 & 0.010 & 0.641 & 0.680 & 30048 & 21330 & 1.00 \\
        $\psi_{1,1}$ & 0.928 & 0.240 & 0.530 & 1.466 & 258 & 868 & 1.01 \\
        $\psi_{1,2}$ & 0.262 & 0.022 & 0.220 & 0.305 & 11471 & 16090 & 1.00 \\
        $\psi_{1,3}$ & 0.072 & 0.015 & 0.049 & 0.108 & 10903 & 14986 & 1.00 \\
        $\psi_{1,4}$ & 0.199 & 0.015 & 0.171 & 0.230 & 19095 & 20526 & 1.00 \\
        $\psi_{2,1}$ & 3.156 & 0.473 & 2.316 & 4.149 & 328 & 1017 & 1.02 \\
        $\psi_{2,2}$ & 0.854 & 0.105 & 0.673 & 1.085 & 5243 & 10670 & 1.00 \\
        $\psi_{2,3}$ & 0.585 & 0.074 & 0.459 & 0.748 & 6704 & 12945 & 1.00 \\
        $\psi_{2,4}$ & 0.523 & 0.072 & 0.397 & 0.681 & 8257 & 14377 & 1.00 \\
        \bottomrule
    \end{tabular}
\end{table}

\section{Limitations of the Hybrid Sampler}
\label{sec:Limitation}

While the proposed hybrid sampler in principle handles both ordinal responses and negative binomially distributed responses, it is inefficient for the former and not straightforward to apply to the latter. We here explain why.

\subsection{Ordinal Responses}

While it is possible to formulate Pólya-Gamma models for ordinal data \citep{jimenezSequentialExploratoryDiagnostic2023}, the probit model \citep{mccullaghRegressionModelsOrdinal1980} is more common and convenient, so we limit our discussion to this case.

Let $y_{itj} \in \{1, 2, \dots, C\}$ ($j=1,\dots,U$) be an ordered categorical response with $C$ categories. The corresponding latent response is mapped to the observed response via a set of strictly increasing threshold parameters $\bm{\theta}_{j} = (\tau_{0,j}, \tau_{1,j}, \dots, \tau_{C,j})$, where $\tau_{0,j} = -\infty$ and $\tau_{C,j} = \infty$. We assume that
\begin{equation}
\label{eq:OrdinalResponse}
    y_{itj} = c \quad \text{if} \quad \tau_{c-1, j} < \tilde{y}_{itj} \le \tau_{c, j}, \quad i=1,\dots,N, \quad t=1,\dots, T.
\end{equation}
As pointed out by \citet{cowlesAcceleratingMonteCarlo1996} in the context of Gibbs sampling, the constraint \eqref{eq:OrdinalResponse} implies that the conditional posterior distribution of the threshold parameter $\tau_{c,j}$ is uniform over the interval between the highest latent response for which $y_{itj}=c$ and the lowest latent response for which $y_{itj}=c+1$, 
\begin{equation}
\label{eq:OrdinalConstraint}
    \left[ \text{max}_{i,t}\left(\tilde{y}_{itj} : y_{itj} = c \right), 
    \text{min}_{i,t}\left(\tilde{y}_{itj} : y_{itj} = c+1\right) \right], \quad i=1,\dots,N, \quad  t=1,\dots,T.
\end{equation}
In the NUTS part of the hybrid sampler, where $\tau_{c,j}$ is updated, the exact likelihood contribution of the latent responses is an indicator function which is one if $\tau_{c,j}$ satisfies \eqref{eq:OrdinalConstraint} and zero otherwise. With a large number of observations, as is the case with intensive longitudinal data, the interval will be very narrow, and NUTS will be forced to take extremely small steps to reach its target acceptance probability. Consequently, it will be very inefficient. 

Pure NUTS, on the other hand, avoids the whole problem because it does not use a latent response formulation. For the five-indicator AR(1) model introduced in Section \ref{sec:five-ind-ar1} we confirmed this empirically by introducing ordinal responses with $C=3$ categories. Table \ref{tab:example3-summary} shows diagnostics averaged over five runs with $N=20$, $T=50$, and otherwise the same parametrization as in Section \ref{sec:example2-bernoulli-participant-varying}. The number of iterations was 5,000 for both the hybrid sampler and NUTS, and the first 1,000 iterations were discarded as burn-in. In this case NUTS' performance is excellent, whereas the hybrid sampler is very inefficient. The same pattern was shown in experiments with $C=5$ and $C=7$ categories.

\begin{table}[htpb]
    \centering
    \caption{Benchmark results for the ordinal outcome AR(1) model with participant-varying dynamics. Averages over 5 runs}
    \label{tab:example3-summary}
    \begin{tabular}{lrrrrrr}
        \toprule
        \textbf{Algorithm} & \textbf{Wall time (s)} & \textbf{Bulk-ESS} & \textbf{Tail-ESS} & \textbf{Bulk-ESS/s} & \textbf{Tail-ESS/s} & $\hat{\boldsymbol{R}}$ \\
        \midrule
        \textbf{Hybrid (NumPyro)} & 181.7 & 5 & 11 & 0.03 & 0.07 & 2.880 \\
        \textbf{NUTS (NumPyro)} & 284.9 & 2265 & 3197 & 7.75 & 11.78 & 1.000 \\
        \bottomrule
    \end{tabular}
\end{table}

This issue cannot be resolved by updating $\tau_{c,j}$ in the Gibbs step of the hybrid sampler, as this again will lead to highly auto-correlated chains. Following \citet{cowlesAcceleratingMonteCarlo1996} one could instead jointly propose latent responses and their corresponding thresholds in a Metropolis step. However, Metropolis steps require manual tuning of proposal distributions, making automated inference challenging. We conclude that developing scalable algorithms for DSEMs with ordinal data remains an open challenge.

\subsection{Negative Binomial Responses}

Negative binomial models are convenient for representing overdispersed count data, e.g., when the variance of Poisson distributed counts is larger than the rate $\lambda$ \citep{gardnerRegressionAnalysesCounts1995}. Letting $y$ be the count, $r$ the number of failures (the dispersion parameter), and $p$ the success probability, the negative binomial likelihood can be represented as a scale mixture of Gaussians with respect to a latent Pólya-Gamma random variable $\omega \sim \mathcal{PG}\{y+r, y^{*}\}$ where $y^{*}$ is the linear predictor \citep{polsonBayesianInferenceLogistic2013}. This introduces a new parameter $r$ which needs to be updated in the hybrid sampler. 

Since $r$ is a continuous parameter, the shape parameter $y+r$ is generally not an integer, so $\omega$ cannot be obtained as a sum of independent $\mathcal{PG}\{1, y^{*}\}$ draws as with binomial responses. This is not a problem in itself, since samplers for non-integer shape parameters exist \citep{windleSamplingPolyaGammaRandom2014} and are implemented in \textsc{polyagamma} \citep{bleki_polyagamma}. The difficulty lies in updating $r$. Conditionally on $\omega$, $r$ also enters the augmented likelihood through the density of $\omega$ under the $\mathcal{PG}\{y+r, 0\}$ distribution, which has no closed form and is given by an infinite alternating series \citep{polsonBayesianInferenceLogistic2013}. This term, and its gradient with respect to $r$, would hence have to be included in the target density of the NUTS step. Alternatively, $r$ could be updated with $\omega$ integrated out, using the negative binomial likelihood conditionally on the linear predictor, in a separate step between the FFBS step and the sampling of $\omega$, in line with the ordering requirements of partially collapsed Gibbs samplers \citep{vandykPartiallyCollapsedGibbs2008}. We have not explored whether either approach would be efficient in practice.

\section{Discussion}
\label{sec:Discussion}

We have presented a hybrid NUTS-Gibbs sampler for DSEMs with binomial responses. It naturally supports both the probit link through the use of latent Gaussian responses and the logit link through the use of Pólya-Gamma auxiliary variables. Two key features make the algorithm efficient for estimating DSEMs. First, by analytically integrating over the within-level latent variables using the Kalman filter, it allows NUTS to target a marginal posterior with $\mathcal{O}(N+T)$ parameters, rather than $\mathcal{O}(N \cdot T)$ as would otherwise be the case. This way, the algorithm retains all the well-known advantages of NUTS \citep{betancourtConceptualIntroductionHamiltonian2018} for the high-dimensional parameter space consisting of participant-level parameters, timepoint-level parameters, and population-level parameters. Second, the Gibbs step in which the $\mathcal{O}(N \cdot T)$ latent response variables are sampled, is very efficient. In particular, all participants can be treated completely independently, making this an ideal case for parallelization when the number of participants is large. Furthermore, by confining Gibbs sampling to latent responses only, we avoid the high autocorrelation and convergence problems typical in Gibbs samplers for high-dimensional models \citep{parkImprovingGibbsSampler2022}. 

Our simulation experiments confirmed these theoretical expectations. For the AR(1) model with participant-invariant dynamics (Figure \ref{fig:example1-efficiency}), the hybrid sampler was more efficient than all competing algorithms for all conditions considered, and it made DSEM estimation feasible both in the case of a very high number of timepoints ($N=20$ and $T=2500$) and in the case of a very high number of participants ($N=500$ and $T=50$). For the AR(1) model with participant-varying dynamics, a Gibbs sampler was not possible due to lack of conjugacy. In this case, pure NUTS was more efficient than the hybrid sampler in the $N<T$ case (Figure \ref{fig:example2-efficiency}, left), while the hybrid sampler was more efficient in the $N>T$ case (Figure \ref{fig:example2-efficiency}, right). This result matches well with the scaling discussed above: since both the Kalman filter inside the NUTS step of the hybrid sampler as well as the Gibbs step can be performed independently across participants, but not across timepoints, we expect the hybrid sampler to be particularly competitive in the large $N$ case. Even when running four chains in parallel, \textsc{NumPyro} uses the XLA compiler to efficiently vectorize these operations across the $N$ participants in each chain.

In the more challenging simulation experiment with a VAR(1) model with three latent traits and participant-varying dynamics, the hybrid sampler's efficiency gain compared to pure NUTS was even more pronounced (Table \ref{tab:example4-summary}). We expect this to be the case also for other DSEMs with more than a single latent trait, due to how the hybrid sampler always marginalizes out these traits in the NUTS step, whereas pure NUTS always has to sample them. Testing this in actual applications, along with large-scale simulation studies, is an important avenue for further research.

Making DSEMs computationally feasible for even bigger data than what was considered here, likely requires the use of graphical processing units (GPUs). GPUs are optimized for highly parallelizable linear algebra, of which the Kalman filter operations in the hybrid sampler are an excellent example. In contrast, in pure NUTS the leapfrog integrator would introduce a bottleneck given that Hamiltonian dynamics for the $\mathcal{O}(N \cdot T)$ parameters would need to be simulated sequentially. Implementing this algorithm in \textsc{NumPyro} is also an advantage in this context. Because \textsc{NumPyro} is powered by the JAX library and the XLA compiler, the exact same model code can be executed on standard central processing units (CPUs), GPUs, or tensor processing units (TPUs) \citep{jouppiInDatacenterPerformanceAnalysis2017} simply by changing the backend execution target. This allows researchers to seamlessly transition from prototyping DSEMs on personal computers to analyzing massive datasets on high-performance computing clusters without having to rewrite any underlying code. Furthermore, the temporal parallelization algorithm of \citet{sarkkaTemporalParallelizationBayesian2021} could potentially reduce the time required for running the Kalman filter for each individual participant from $\mathcal{O}(T)$ for the algorithm outlined in Section \ref{sec:KalmanFilterFFBS} to $\mathcal{O}(\log T)$.

Sampling from the Pólya-Gamma distribution, however, poses a potential challenge to GPU acceleration. As mentioned, the rejection sampler proposed by \citet{polsonBayesianInferenceLogistic2013}, which we used for all Bernoulli responses, has an acceptance probability which is lower-bounded by $0.99919$. If we conservatively assume that the actual probability equals this lower bound, for a model with $U$ Bernoulli distributed responses with the logit link, this means that the probability of experiencing at least one rejection when processing a single participant in a single iteration of the Gibbs sampler is $1 - 0.99919^{U T}$. For example, with $U=3$ and $T=100$, the probability is $0.22$. While this poses no problem for CPUs, GPUs are most efficient when the parallel operations take about the same amount of time to complete. Rejections, on the other hand, require new proposals to be sampled, and since rejections will happen to an unequal extent across participants, this causes branch divergence, forcing accepted threads to sit idle and wait while rejected threads resample. The extent to which this problem is a real bottleneck will await empirical evaluation. However, one potential solution would be to approximate the Pólya-Gamma distribution by truncating the infinite sum \eqref{eq:PolyaGammaSum} defining its density at some high number of terms. 

The hybrid NUTS-Gibbs sampler proposed in this paper, along with the NUTS-Kalman algorithm of \citet{sorensenEfficientBayesianEstimation2026}, makes DSEM estimation with big data and complex models considerably more feasible than existing algorithms, as long as the responses are either Bernoulli, binomially or normally distributed. Many other response types are of high practical importance, however. Due to the widespread use of Likert scales in psychological assessment, analysis of ordinal data is particularly relevant \citep{jimenezSequentialExploratoryDiagnostic2023}, but count data is also common, e.g., from ecological momentary assessments of E-cigarette use \citep{buuEcologicalMomentaryAssessment2023} or of the number of alcohol-related problems which occurred during a day \citep{gaherExperienceSamplingStudy2014}. Properly modeling such data requires extending the Kalman filter marginalization approach to general exponential family responses, as well as to the negative binomial distribution. As we pointed out in Section \ref{sec:Limitation}, while the hybrid sampler in principle allows ordinal or negative binomial responses, it is inefficient for ordinal responses and not straightforward to apply to negative binomial responses. Keeping the algorithm consistent while marginalizing over latent states likely requires replacing the Kalman filter with particle filters, as in pseudo-marginal MCMC \citep{andrieuParticleMarkovChain2010}. Unfortunately, such algorithms are known to be extremely slow. We believe a more viable approach is to extend the NUTS-Kalman algorithm of \citet{sorensenEfficientBayesianEstimation2026} such that it approximately marginalizes over latent states using an unscented Kalman filter \citep{julierNewExtensionKalman1997}. An important question when comparing such an approximate algorithm to an alternative which targets the correct posterior is how accurate estimates of the posterior distribution one can obtain with a given compute budget. This is an important question for further research in order to maximize the usability of these samplers in applied research contexts. 

\printbibliography

\startsupplement{1}

This supplement derives the adjustment of Theorem 1 in \citet{sorensenEfficientBayesianEstimation2026} to the within-level model defined in terms of the centered linear predictor $\bm{y}_{1,it}^{*}$.

\begin{definition}[Strictly Lagged Polynomial Matrices]
    \label{def:strictly-lagged}
    For any polynomial matrix $\bm{P}(\mathcal{L}) = \sum_{l=0}^{L} \bm{P}_{l} \mathcal{L}^{l}$, we define the strictly lagged polynomial matrix $\bm{P}^{*}(\mathcal{L})$ as the polynomial obtained by removing the contemporaneous (zero-lag) coefficient $\bm{P}_{0}$:
    \begin{equation}
        \bm{P}^{*}(\mathcal{L}) = \bm{P}(\mathcal{L}) - \bm{P}_{0} = \sum_{l=1}^{L} \bm{P}_{l} \mathcal{L}^{l}.
    \end{equation}
\end{definition}

\begin{definition}[Coefficient Extraction]
    \label{def:coeff-extraction}
    Any time-varying polynomial operator $\bm{\Phi}_{it}(\mathcal{L})$ can be uniquely written in the left-canonical form $\bm{\Phi}_{it}(\mathcal{L}) = \sum_{k} \bm{D}_{k,it} \mathcal{L}^k$. We define the extraction operator $[\cdot]_k$ such that $[\bm{\Phi}_{it}(\mathcal{L})]_k \equiv \bm{D}_{k,it}$.
\end{definition}

The adjusted version of the theorem is:

\begin{theorem}
    \label{th:Theorem1_modified}
    Let $\bm{A}_{0it} = (\bm{I} - \bm{R}_{0it})^{-1}$ and $\bm{\Xi}_{0it} = (\bm{I} - \bm{B}_{1,0it})^{-1}$, which exist since $\bm{R}_{0it}$ and $\bm{B}_{1,0it}$ are strictly lower triangular, and assume that $\bm{I} - \bm{\Xi}_{0it} \bm{Q}_{0it} \bm{A}_{0it} \bm{\Lambda}_{1,0it}$ is nonsingular for all $i$ and $t$. Then, for maximum lag $L \ge 1$, the within-level model linking the latent states to the linear predictor is exactly equivalent to the state space transition model
    \begin{equation}
        \label{eq:StateSpaceModelTheorem1}
        \tilde{\bm{\eta}}_{1,i,t+1} = \bm{T}_{it} \tilde{\bm{\eta}}_{1,it} + \bm{c}_{it} + \bm{w}_{it}, \quad \bm{w}_{it} \sim \mathcal{N}\left(\bm{0}, \bm{W}_{it}\right),
    \end{equation}
    where the augmented state vector tracks the latent states and continuous linear predictors, $\tilde{\bm{\eta}}_{1,it} = [\bm{\eta}_{1,it}^{T}, \dots, \bm{\eta}_{1,i,t-L+1}^{T}, (\bm{y}_{1,it}^{*})^{T}, \dots, (\bm{y}_{1,i,t-L+1}^{*})^{T}]^{T}$. The contemporaneous continuous linear predictor is extracted via $\bm{y}_{1,it}^{*} = \bm{Z}_{it} \tilde{\bm{\eta}}_{1,it}$, where the measurement matrix is $\bm{Z}_{it} = [\bm{0}_{U \times L V_1} \quad \bm{I}_{U \times U} \quad \bm{0}_{U \times (L-1)U}]$. The transition matrix is
    \begin{equation*}
        \bm{T}_{it} =
        \begin{bmatrix}
            \bm{T}_{it}^{(1,1)}                                                            & \bm{T}_{it}^{(1,2)}                                                      \\
            \begin{bmatrix} \bm{I}_{(L-1)V_1} & \bm{0}_{(L-1)V_1 \times V_1} \end{bmatrix} & \bm{0}                                                                   \\
            \bm{T}_{it}^{(3,1)}                                                            & \bm{T}_{it}^{(3,2)}                                                      \\
            \bm{0}                                                                         & \begin{bmatrix} \bm{I}_{(L-1)U} & \bm{0}_{(L-1)U \times U} \end{bmatrix}
        \end{bmatrix} ,
    \end{equation*}
    where the block-columns for $k=1,\dots,L$ are defined by
    \begin{align*}
        \{\bm{T}_{it}^{(1,1)}\}_{k} & = \bm{M}_{1,i,t+1} [\bm{\mathcal{P}}^{\eta}_{i,t+1}(\mathcal{L})]_k, &
        \{\bm{T}_{it}^{(1,2)}\}_{k} & = \bm{M}_{1,i,t+1} [\bm{\mathcal{P}}^{y^*}_{i,t+1}(\mathcal{L})]_k       \\
        \{\bm{T}_{it}^{(3,1)}\}_{k} & = \bm{N}_{1,i,t+1} [\bm{\mathcal{Q}}^{\eta}_{i,t+1}(\mathcal{L})]_k, &
        \{\bm{T}_{it}^{(3,2)}\}_{k} & = \bm{N}_{1,i,t+1} [\bm{\mathcal{Q}}^{y^*}_{i,t+1}(\mathcal{L})]_k.
    \end{align*}
    The matrices $\bm{M}_{1,i,t+1}$ and $\bm{N}_{1,i,t+1}$ are defined in \eqref{eq:M1it} and \eqref{eq:N1it}, the composite polynomials $\bm{\mathcal{P}}^{\eta}_{i,t+1}(\mathcal{L})$, $\bm{\mathcal{P}}^{y^*}_{i,t+1}(\mathcal{L})$, $\bm{\mathcal{Q}}^{\eta}_{i,t+1}(\mathcal{L})$, and $\bm{\mathcal{Q}}^{y^*}_{i,t+1}(\mathcal{L})$ are defined in \eqref{eq:P_poly} and \eqref{eq:Q_poly}, and the intercept vector $\bm{c}_{it}$ and process noise covariance matrix $\bm{W}_{it}$ are defined in \eqref{eq:c_intercept} and \eqref{eq:w_process_noise}.
\end{theorem}

\begin{proof}
    Isolating the lag-$0$ components on the left-hand side, the coupled within-level model for the discrete response case can be rewritten as
    \begin{align}
        \label{eq:ystar_separated}
        \left(\bm{I} - \bm{R}_{0it}\right) \bm{y}_{1,it}^{*} - \bm{\Lambda}_{1,0it} \bm{\eta}_{1,it} & = \bm{\nu}_{1,it} + \bm{K}_{1,it} \bm{X}_{1,it} + \bm{\Lambda}^{*}_{1,it}(\mathcal{L}) \bm{\eta}_{1,it} + \bm{R}^{*}_{it}(\mathcal{L}) \bm{y}_{1,it}^{*} \\
        \label{eq:eta_separated}
        \left(\bm{I} - \bm{B}_{1,0it}\right) \bm{\eta}_{1,it} - \bm{Q}_{0it} \bm{y}_{1,it}^{*}       & = \bm{\alpha}_{1,it} + \bm{\Gamma}_{1,it} \bm{X}_{1,it} + \bm{B}^{*}_{1,it}(\mathcal{L}) \bm{\eta}_{1,it} + \bm{Q}^{*}_{it}(\mathcal{L}) \bm{y}_{1,it}^{*} + \bm{\xi}_{1,it}.
    \end{align}
    Note that \eqref{eq:ystar_separated} differs from equation (A.2) in \citet{sorensenEfficientBayesianEstimation2026} in that it contains no explicit measurement error term $\bm{\epsilon}_{1,it}$.
    
    With $\bm{A}_{0it}$ and $\bm{\Xi}_{0it}$ defined as in Theorem~\ref{th:Theorem1_modified}, we assemble the left-hand sides of \eqref{eq:ystar_separated} and \eqref{eq:eta_separated} into a $2 \times 2$ block matrix structure and compute its inverse to isolate the contemporaneous states:
    \begin{equation}
        \label{eq:JointSystem}
        \begin{bmatrix} \bm{\eta}_{1,it} \\ \bm{y}_{1,it}^{*} \end{bmatrix}
        =
        \begin{bmatrix} \bm{\Xi}_{0it}^{-1} & -\bm{Q}_{0it} \\ -\bm{\Lambda}_{1,0it} & \bm{A}_{0it}^{-1} \end{bmatrix}^{-1}
        \begin{bmatrix} \text{RHS}_{\eta, it} \\ \text{RHS}_{y^*, it} \end{bmatrix},
    \end{equation}
    where $\text{RHS}_{\eta, it}$ and $\text{RHS}_{y^*, it}$ represent the respective right-hand sides of \eqref{eq:eta_separated} and \eqref{eq:ystar_separated}. Block-matrix inversion via the Schur complement yields
    \begin{equation*}
        \begin{bmatrix}
            \bm{\Xi}_{0it}^{-1}   & -\bm{Q}_{0it}     \\
            -\bm{\Lambda}_{1,0it} & \bm{A}_{0it}^{-1}
        \end{bmatrix}^{-1}
        = \begin{bmatrix}
            \bm{M}_{1,it}                                     & \bm{M}_{1,it} \bm{Q}_{0it} \bm{A}_{0it} \\
            \bm{N}_{1,it} \bm{\Lambda}_{1,0it} \bm{\Xi}_{0it} & \bm{N}_{1,it}
        \end{bmatrix},
    \end{equation*}
    where
    \begin{align}
        \bm{M}_{1,it} & = \left\{ \bm{I} - \bm{\Xi}_{0it} \bm{Q}_{0it} \bm{A}_{0it} \bm{\Lambda}_{1,0it} \right\}^{-1} \bm{\Xi}_{0it} \label{eq:M1it} \\
        \bm{N}_{1,it} & = \left\{ \bm{I} - \bm{A}_{0it} \bm{\Lambda}_{1,0it} \bm{\Xi}_{0it} \bm{Q}_{0it}  \right\}^{-1} \bm{A}_{0it}.
        \label{eq:N1it}
    \end{align}
    Both inverses in \eqref{eq:M1it} and \eqref{eq:N1it} exist under the assumption of the theorem, since $\det(\bm{I} - \bm{X}\bm{Y}) = \det(\bm{I} - \bm{Y}\bm{X})$ with $\bm{X} = \bm{\Xi}_{0it} \bm{Q}_{0it}$ and $\bm{Y} = \bm{A}_{0it} \bm{\Lambda}_{1,0it}$. The assumption is equivalent to the block matrix in \eqref{eq:JointSystem} being invertible, that is, to the contemporaneous system having a unique solution, and it holds automatically when $\bm{Q}_{0it} = \bm{0}$.

    By advancing the indices of \eqref{eq:JointSystem} to $t+1$ and distributing the block-inverse across the right-hand sides, we identify the composite transition polynomials. For the structural update equation ($\bm{\eta}_{1,i,t+1}$), we find
    \begin{equation}
        \begin{aligned}
            \bm{\mathcal{P}}^{\eta}_{i,t+1}(\mathcal{L}) & = \bm{B}^{*}_{1,i,t+1}(\mathcal{L}) + \bm{Q}_{0,i,t+1} \bm{A}_{0,i,t+1} \bm{\Lambda}^{*}_{1,i,t+1}(\mathcal{L}) \\
            \bm{\mathcal{P}}^{y^*}_{i,t+1}(\mathcal{L})    & = \bm{Q}^{*}_{i,t+1}(\mathcal{L}) + \bm{Q}_{0,i,t+1} \bm{A}_{0,i,t+1} \bm{R}^{*}_{i,t+1}(\mathcal{L}).
        \end{aligned}
        \label{eq:P_poly}
    \end{equation}
    For the linear predictor update equation ($\bm{y}^{*}_{1,i,t+1}$), grouping the lagged state elements identifies
    \begin{equation}
        \begin{aligned}
            \bm{\mathcal{Q}}^{\eta}_{i,t+1}(\mathcal{L}) & = \bm{\Lambda}^{*}_{1,i,t+1}(\mathcal{L}) + \bm{\Lambda}_{1,0,i,t+1} \bm{\Xi}_{0,i,t+1} \bm{B}^{*}_{1,i,t+1}(\mathcal{L}) \\
            \bm{\mathcal{Q}}^{y^*}_{i,t+1}(\mathcal{L})    & = \bm{R}^{*}_{i,t+1}(\mathcal{L}) + \bm{\Lambda}_{1,0,i,t+1} \bm{\Xi}_{0,i,t+1} \bm{Q}^{*}_{i,t+1}(\mathcal{L}).
        \end{aligned}
        \label{eq:Q_poly}
    \end{equation}
    Extracting the $k$-th order coefficients generates the transition sub-blocks $\bm{T}_{it}^{(1,1)}$, $\bm{T}_{it}^{(1,2)}$, $\bm{T}_{it}^{(3,1)}$, and $\bm{T}_{it}^{(3,2)}$. Distributing the block-inverse matrices across the exogenous deterministic variables yields the intercepts
    \begin{align*}
        \bm{c}_{it}^{(1)} & = \bm{M}_{1,i,t+1} \left\{ \bm{\alpha}_{1,i,t+1} + \bm{\Gamma}_{1,i,t+1} \bm{X}_{1,i,t+1} + \bm{Q}_{0,i,t+1}\bm{A}_{0,i,t+1} \left( \bm{\nu}_{1,i,t+1} + \bm{K}_{1,i,t+1} \bm{X}_{1,i,t+1} \right) \right\}            \\
        \bm{c}_{it}^{(3)} & = \bm{N}_{1,i,t+1} \left\{ \bm{\nu}_{1,i,t+1} + \bm{K}_{1,i,t+1} \bm{X}_{1,i,t+1} + \bm{\Lambda}_{1,0,i,t+1}\bm{\Xi}_{0,i,t+1} \left( \bm{\alpha}_{1,i,t+1} + \bm{\Gamma}_{1,i,t+1} \bm{X}_{1,i,t+1} \right) \right\},
    \end{align*}
    giving the full augmented intercept vector
    \begin{equation}
        \label{eq:c_intercept}
        \bm{c}_{it} = \begin{bmatrix}
            (\bm{c}_{it}^{(1)})^{T} & \bm{0}_{(L-1)V_1}^{T} & (\bm{c}_{it}^{(3)})^{T} & \bm{0}_{(L-1)U}^{T}
        \end{bmatrix}^{T}.
    \end{equation}
    Finally, because the observation-level noise $\bm{\epsilon}_{1,i,t+1}$ is absent from the structural transitions, distributing the inverse matrices over the stochastic errors isolates the serially independent process noise components:
    \begin{align}
        \bm{w}_{it}^{(1)} & = \bm{M}_{1,i,t+1} \bm{\xi}_{1,i,t+1} \label{eq:w1_noise} \\
        \bm{w}_{it}^{(3)} & = \bm{N}_{1,i,t+1} \bm{\Lambda}_{1,0,i,t+1} \bm{\Xi}_{0,i,t+1} \bm{\xi}_{1,i,t+1}. \label{eq:w3_noise}
    \end{align}
 Taking the variances and cross-covariances of \eqref{eq:w1_noise} and \eqref{eq:w3_noise} generates the exact, simplified covariance components:
    \begin{align*}
        \bm{W}_{it}^{(1,1)} & = \bm{M}_{1,i,t+1} \bm{\Psi}_{1,i,t+1} \bm{M}_{1,i,t+1}^{T}                     \\
        \bm{W}_{it}^{(3,3)} & = \bm{N}_{1,i,t+1} \bm{\Lambda}_{1,0,i,t+1} \bm{\Xi}_{0,i,t+1} \bm{\Psi}_{1,i,t+1} \bm{\Xi}_{0,i,t+1}^{T} \bm{\Lambda}_{1,0,i,t+1}^{T} \bm{N}_{1,i,t+1}^{T} \\
        \bm{W}_{it}^{(1,3)} & = \bm{M}_{1,i,t+1} \bm{\Psi}_{1,i,t+1} \bm{\Xi}_{0,i,t+1}^{T} \bm{\Lambda}_{1,0,i,t+1}^{T} \bm{N}_{1,i,t+1}^{T},
    \end{align*}
    with $\bm{W}_{it}^{(3,1)} = \left( \bm{W}_{it}^{(1,3)} \right)^{T}$. 
    
    Because the strict state lag elements transition deterministically, they contain no process variance. This enforces a sparse, block-diagonal covariance mapping corresponding to our $L$-lag augmented state vector space
    \begin{equation}
        \label{eq:w_process_noise}
        \bm{W}_{it} =
        \begin{bmatrix}
            \bm{W}_{it}^{(1,1)} & \bm{0} & \bm{W}_{it}^{(1,3)} & \bm{0} \\
            \bm{0}              & \bm{0} & \bm{0}              & \bm{0} \\
            \bm{W}_{it}^{(3,1)} & \bm{0} & \bm{W}_{it}^{(3,3)} & \bm{0} \\
            \bm{0}              & \bm{0} & \bm{0}              & \bm{0}
        \end{bmatrix}.
    \end{equation}
    This formulation guarantees exact equivalence with the LG-SSM dynamics for the augmented state equations. 
\end{proof}

\startsupplement{2}

This supplement contains additional details from the simulation experiments reported in the main paper.

\section{Bernoulli Response Five-Indicator AR(1) Models}

\subsection{Participant-Invariant Dynamics}

Figure \ref{fig:example1-bulk-efficiency-logit} shows bulk efficiency with the logit link. Figures \ref{fig:example1-tail-efficiency-probit} and \ref{fig:example1-tail-efficiency-logit} show tail efficiency with the probit and logit links, respectively. Tables \ref{tab:example1-efficiency_probit} and \ref{tab:example1-efficiency_logit} show the bulk efficiency gain of the hybrid sampler compared to the other algorithms for the probit and logit links, respectively. Tables \ref{tab:example1-waittime_probit} and \ref{tab:example1-waittime_logit} show estimated times one will have to wait to reach an effective sample size of $1{,}000$. Table \ref{tab:example1-iterations} shows the number of iterations run for all algorithms in all conditions. Tables \ref{tab:example1-diag_probit} and \ref{tab:example1-diag_logit} show MCMC diagnostics for all settings. Tables \ref{tab:example1-bias_probit_N50_T50} and \ref{tab:example1-bias_logit_N50_T50} show the average posterior means together with biases for all parameters and algorithms in the $N=50$ and $T=50$ case, strongly suggesting that our implementations correctly target the posterior distribution. Similar tables for other $N$ and $T$ can be generated with the script \texttt{code/example1/make\_posterior\_tables.py} in the OSF repository, and they all give very similar numbers.

The hybrid sampler and the NUTS samplers were run with a target acceptance probability during warm-up of $0.80$ and $10$ as the maximum treedepth. All five algorithms used exactly the same initial values. In terms of the original parametrization, these were $\phi = 0$, $\psi_{1}^{2}=1$, $\psi_{2}^{2}=1$, $\bm{\nu}=\bm{0}$, $\bm{\lambda}_{1}=\bm{0}$, and $\bm{\lambda}_{2}=\bm{0}$. 

For all algorithms except the \textsc{JAGS} implementation, the initial value of the within-level latent state was sampled from a zero-mean normal distribution with variance equal to the marginal variance of the AR(1) process,
\begin{equation*}
    \eta_{1,i1} \sim \mathcal{N}\left(0, \frac{\psi_{1}^{2}}{1-\phi^{2}}\right).
\end{equation*}
Since \textsc{JAGS} would not be able to derive closed-form conditional posteriors with such a non-linear dependence, we here instead used $\eta_{1,it} \sim \mathcal{N}(0,1)$ to prevent it from falling back on Metropolis-Hastings sampling.

\begin{figure}
    \includegraphics[width=\linewidth]{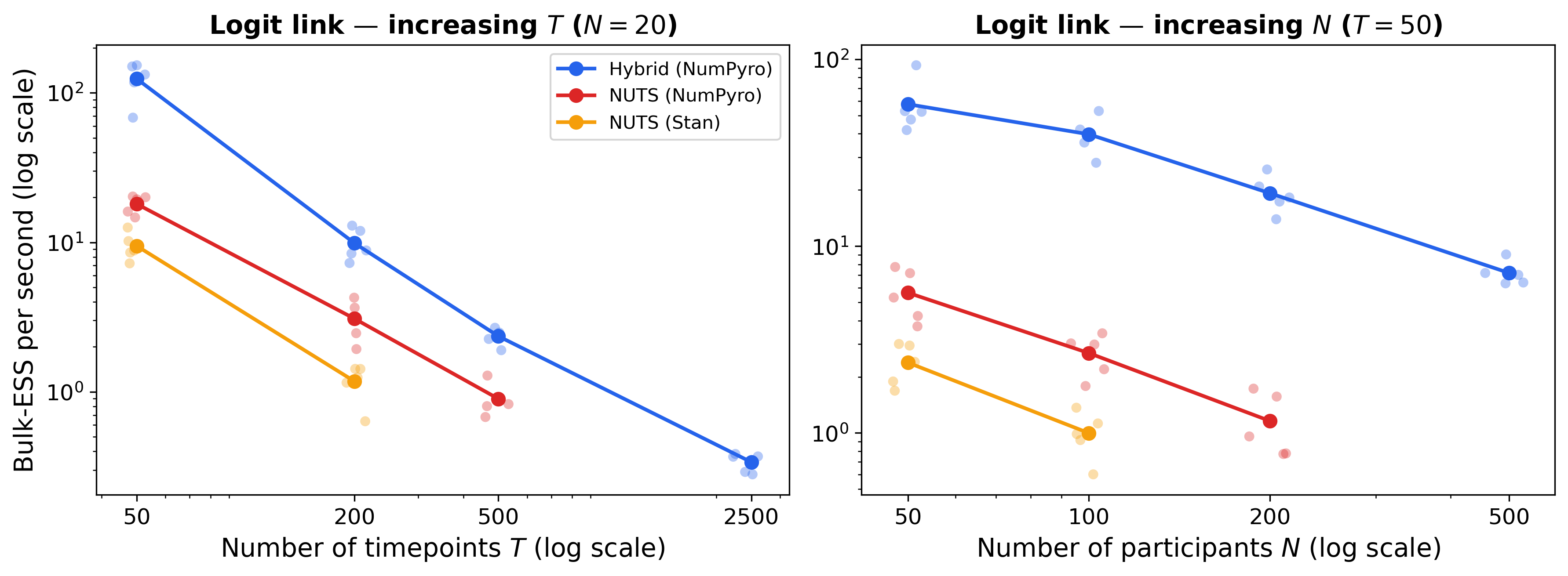}
    \caption{Bulk efficiency plot for the five-indicator AR(1) model with participant-invariant dynamics with logit link.}
    \label{fig:example1-bulk-efficiency-logit}
\end{figure}

\begin{figure}
    \includegraphics[width=\linewidth]{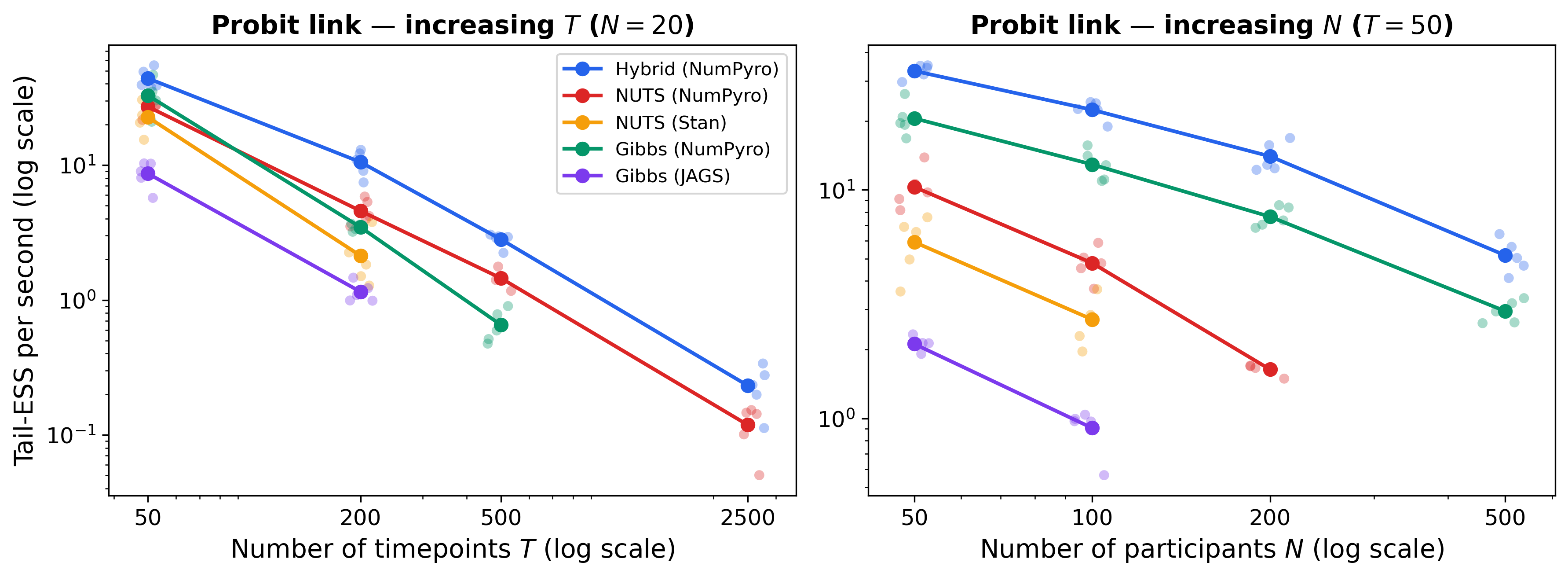}
    \caption{Tail efficiency plot for the five-indicator AR(1) model with participant-invariant dynamics with probit link.}
    \label{fig:example1-tail-efficiency-probit}
\end{figure}

\begin{figure}
    \includegraphics[width=\linewidth]{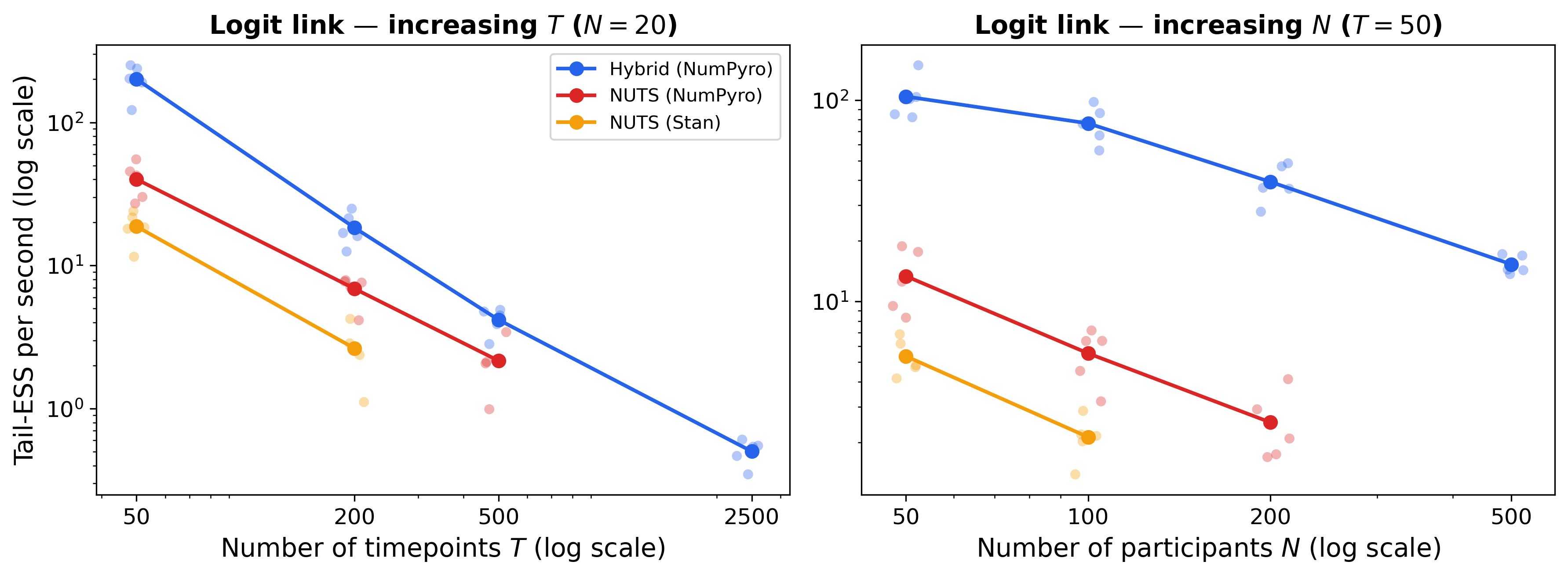}
    \caption{Tail efficiency plot for the five-indicator AR(1) model with participant-invariant dynamics with logit link.}
    \label{fig:example1-tail-efficiency-logit}
\end{figure}

\begin{table}[htbp]
\centering
\caption{Efficiency gain of the hybrid sampler for the probit  link.}
\label{tab:example1-efficiency_probit}
\small
\begin{tabular}{ccrrrr}
\toprule
$N$ & $T$ & NUTS (NumPyro) & NUTS (Stan) & Gibbs (NumPyro) & Gibbs (JAGS) \\
\midrule
20 & 50 & 1.8 & 2.5 & 1.7 & 6.8 \\
20 & 200 & 3.2 & 5.1 & 3.3 & 11.6 \\
20 & 500 & 2.4 & --- & 5.5 & --- \\
20 & 2500 & 2.3 & --- & --- & --- \\
50 & 50 & 4.0 & 6.2 & 1.7 & 16.8 \\
100 & 50 & 4.6 & 9.2 & 1.9 & 27.7 \\
200 & 50 & 8.6 & --- & 1.9 & --- \\
500 & 50 & --- & --- & 1.9 & --- \\
\bottomrule
\end{tabular}
\end{table}

\begin{table}[htbp]
\centering
\caption{Efficiency gain of the hybrid sampler for the logit  link.}
\label{tab:example1-efficiency_logit}
\small
\begin{tabular}{ccrrrr}
\toprule
$N$ & $T$ & NUTS (NumPyro) & NUTS (Stan) & Gibbs (NumPyro) & Gibbs (JAGS) \\
\midrule
20 & 50 & 6.9 & 13.1 & --- & --- \\
20 & 200 & 3.2 & 8.4 & --- & --- \\
20 & 500 & 2.6 & --- & --- & --- \\
50 & 50 & 10.2 & 24.2 & --- & --- \\
100 & 50 & 14.8 & 39.8 & --- & --- \\
200 & 50 & 16.6 & --- & --- & --- \\
\bottomrule
\end{tabular}
\end{table}

\begin{table}[htbp]
\centering
\caption{Estimated wall-clock time to obtain an effective sample size of 1000 for the probit link.}
\label{tab:example1-waittime_probit}
\small
\begin{tabular}{ccrclr}
\toprule
$N$ & $T$ & Hybrid & & Second-best algorithm & Time \\
\midrule
20 & 50 & 40\,s & & Gibbs (NumPyro) & 1.1\,min \\
20 & 200 & 2.9\,min & & NUTS (NumPyro) & 9.4\,min \\
20 & 500 & 10.8\,min & & NUTS (NumPyro) & 26.0\,min \\
20 & 2500 & 2.3\,h & & NUTS (NumPyro) & 5.3\,h \\
50 & 50 & 59\,s & & Gibbs (NumPyro) & 1.7\,min \\
100 & 50 & 1.6\,min & & Gibbs (NumPyro) & 2.9\,min \\
200 & 50 & 2.6\,min & & Gibbs (NumPyro) & 4.8\,min \\
500 & 50 & 6.4\,min & & Gibbs (NumPyro) & 12.3\,min \\
\bottomrule
\end{tabular}
\end{table}

\begin{table}[htbp]
\centering
\caption{Estimated wall-clock time to obtain an effective sample size of 1000 for the logit link.}
\label{tab:example1-waittime_logit}
\small
\begin{tabular}{ccrclr}
\toprule
$N$ & $T$ & Hybrid & & Second-best algorithm & Time \\
\midrule
20 & 50 & 8\,s & & NUTS (NumPyro) & 55\,s \\
20 & 200 & 1.7\,min & & NUTS (NumPyro) & 5.4\,min \\
20 & 500 & 7.1\,min & & NUTS (NumPyro) & 18.5\,min \\
50 & 50 & 17\,s & & NUTS (NumPyro) & 3.0\,min \\
100 & 50 & 25\,s & & NUTS (NumPyro) & 6.2\,min \\
200 & 50 & 52\,s & & NUTS (NumPyro) & 14.4\,min \\
\bottomrule
\end{tabular}
\end{table}

\begin{table}[htbp]
\centering
\caption{Number of iterations (warm-up / sampling) for all conditions and algorithms. Missing entries indicate that the algorithm was not run for that condition.}
\label{tab:example1-iterations}
\small
\begin{tabular}{llccccc}
\toprule
$N$ & $T$ & Hybrid (NumPyro) & NUTS (NumPyro) & NUTS (Stan) & Gibbs (NumPyro) & Gibbs (JAGS) \\
\midrule
20 & 50 & 1,000 / 4,000 & 1,000 / 8,000 & 1,000 / 8,000 & 5,000 / 20,000 & 5,000 / 20,000 \\
20 & 200 & 1,000 / 4,000 & 1,000 / 4,000 & 1,000 / 4,000 & 5,000 / 100,000 & 5,000 / 20,000 \\
20 & 500 & 1,000 / 4,000 & 1,000 / 4,000 & -- & 5,000 / 100,000 & -- \\
50 & 50 & 1,000 / 4,000 & 1,000 / 8,000 & 1,000 / 8,000 & 10,000 / 100,000 & 5,000 / 20,000 \\
100 & 50 & 1,000 / 4,000 & 1,000 / 8,000 & 1,000 / 8,000 & 10,000 / 100,000 & 5,000 / 20,000 \\
200 & 50 & 1,000 / 4,000 & 1,000 / 8,000 & -- & 10,000 / 100,000 & -- \\
500 & 50 & 1,000 / 4,000 & -- & -- & 10,000 / 100,000 & -- \\
2500 & 50 & 1,000 / 4,000 & -- & -- & 10,000 / 40,000 & -- \\
\bottomrule
\end{tabular}
\end{table}

\begin{table}[htbp]
\centering
\caption{MCMC diagnostics for the probit link, averaged over 5 datasets.}
\label{tab:example1-diag_probit}
\small
\begin{tabular}{cclrrrrr}
\toprule
$N$ & $T$ & Algorithm & Wall time (s) & Min Bulk-ESS & Min Tail-ESS & Max $\hat{R}$ & Min Bulk-ESS/s \\
\midrule
20 & 50 & Hybrid (NumPyro) & 19 & 480 & 834 & 1.012 & 25 \\
 &  & NUTS (NumPyro) & 69 & 957 & 1874 & 1.006 & 14 \\
 &  & NUTS (Stan) & 89 & 912 & 2017 & 1.006 & 10 \\
 &  & Gibbs (NumPyro) & 26 & 394 & 863 & 1.012 & 15 \\
 &  & Gibbs (JAGS) & 204 & 760 & 1769 & 1.006 & 3.7 \\
\midrule
20 & 200 & Hybrid (NumPyro) & 94 & 536 & 983 & 1.010 & 5.7 \\
 &  & NUTS (NumPyro) & 160 & 281 & 711 & 1.010 & 1.8 \\
 &  & NUTS (Stan) & 277 & 291 & 555 & 1.014 & 1.1 \\
 &  & Gibbs (NumPyro) & 296 & 510 & 1023 & 1.012 & 1.7 \\
 &  & Gibbs (JAGS) & 1083 & 532 & 1243 & 1.010 & 0.49 \\
\midrule
20 & 500 & Hybrid (NumPyro) & 281 & 414 & 769 & 1.012 & 1.5 \\
 &  & NUTS (NumPyro) & 496 & 317 & 718 & 1.018 & 0.64 \\
 &  & Gibbs (NumPyro) & 633 & 176 & 413 & 1.032 & 0.28 \\
\midrule
20 & 2500 & Hybrid (NumPyro) & 2848 & 320 & 612 & 1.016 & 0.12 \\
 &  & NUTS (NumPyro) & 5545 & 291 & 652 & 1.008 & 0.05 \\
\midrule
50 & 50 & Hybrid (NumPyro) & 32 & 543 & 1065 & 1.008 & 17 \\
 &  & NUTS (NumPyro) & 152 & 643 & 1570 & 1.000 & 4.2 \\
 &  & NUTS (Stan) & 227 & 619 & 1344 & 1.008 & 2.7 \\
 &  & Gibbs (NumPyro) & 197 & 1945 & 4046 & 1.000 & 9.9 \\
 &  & Gibbs (JAGS) & 607 & 610 & 1286 & 1.006 & 1.0 \\
\midrule
100 & 50 & Hybrid (NumPyro) & 47 & 495 & 1054 & 1.010 & 11 \\
 &  & NUTS (NumPyro) & 258 & 590 & 1235 & 1.008 & 2.3 \\
 &  & NUTS (Stan) & 505 & 574 & 1351 & 1.008 & 1.1 \\
 &  & Gibbs (NumPyro) & 324 & 1832 & 4186 & 1.002 & 5.7 \\
 &  & Gibbs (JAGS) & 1386 & 527 & 1253 & 1.010 & 0.38 \\
\midrule
200 & 50 & Hybrid (NumPyro) & 77 & 498 & 1078 & 1.010 & 6.5 \\
 &  & NUTS (NumPyro) & 643 & 487 & 1059 & 1.004 & 0.76 \\
 &  & Gibbs (NumPyro) & 564 & 1968 & 4298 & 1.000 & 3.5 \\
\midrule
500 & 50 & Hybrid (NumPyro) & 174 & 457 & 902 & 1.010 & 2.6 \\
 &  & Gibbs (NumPyro) & 1381 & 1876 & 4067 & 1.000 & 1.4 \\
\bottomrule
\end{tabular}
\end{table}

\begin{table}[htbp]
\centering
\caption{MCMC diagnostics for the logit link, averaged over 5 datasets.}
\label{tab:example1-diag_logit}
\small
\begin{tabular}{cclrrrrr}
\toprule
$N$ & $T$ & Algorithm & Wall time (s) & Min Bulk-ESS & Min Tail-ESS & Max $\hat{R}$ & Min Bulk-ESS/s \\
\midrule
20 & 50 & Hybrid (NumPyro) & 21 & 2587 & 4180 & 1.000 & 124 \\
 &  & NUTS (NumPyro) & 33 & 595 & 1314 & 1.012 & 18 \\
 &  & NUTS (Stan) & 63 & 598 & 1181 & 1.006 & 9.5 \\
\midrule
20 & 200 & Hybrid (NumPyro) & 73 & 718 & 1333 & 1.006 & 9.9 \\
 &  & NUTS (NumPyro) & 68 & 202 & 453 & 1.020 & 3.1 \\
 &  & NUTS (Stan) & 170 & 191 & 422 & 1.022 & 1.2 \\
\midrule
20 & 500 & Hybrid (NumPyro) & 235 & 550 & 969 & 1.008 & 2.4 \\
 &  & NUTS (NumPyro) & 201 & 174 & 410 & 1.022 & 0.90 \\
\midrule
20 & 2500 & Hybrid (NumPyro) & 3964 & 1320 & 1938 & 1.002 & 0.34 \\
\midrule
50 & 50 & Hybrid (NumPyro) & 36 & 2046 & 3702 & 1.000 & 58 \\
 &  & NUTS (NumPyro) & 76 & 421 & 996 & 1.010 & 5.6 \\
 &  & NUTS (Stan) & 174 & 416 & 931 & 1.012 & 2.4 \\
\midrule
100 & 50 & Hybrid (NumPyro) & 55 & 2177 & 4197 & 1.000 & 40 \\
 &  & NUTS (NumPyro) & 143 & 382 & 786 & 1.008 & 2.7 \\
 &  & NUTS (Stan) & 393 & 384 & 817 & 1.008 & 1.00 \\
\midrule
200 & 50 & Hybrid (NumPyro) & 96 & 1851 & 3781 & 1.002 & 19 \\
 &  & NUTS (NumPyro) & 318 & 336 & 730 & 1.008 & 1.2 \\
\midrule
500 & 50 & Hybrid (NumPyro) & 231 & 1665 & 3541 & 1.000 & 7.2 \\
\bottomrule
\end{tabular}
\end{table}

\begin{table}[htbp]
\centering
\caption{Posterior recovery with the probit link in the $N = 50$ and $T = 50$ case. Numbers are posterior means averaged across five runs, and numbers in parentheses are average deviations from the true value.}
\label{tab:example1-bias_probit_N50_T50}
\small
\begin{tabular}{lrrrrrr}
\toprule
Parameter & True & Hybrid (NumPyro) & NUTS (NumPyro) & NUTS (Stan) & Gibbs (NumPyro) & Gibbs (JAGS) \\
\midrule
$\phi$ & 0.40 & 0.40 (-0.00) & 0.40 (-0.00) & 0.40 (-0.00) & 0.40 (-0.00) & 0.40 (-0.00) \\
$\psi_1$ & 0.84 & 0.82 (-0.02) & 0.82 (-0.02) & 0.82 (-0.02) & 0.82 (-0.02) & 0.82 (-0.02) \\
$\psi_2$ & 0.50 & 0.49 (-0.01) & 0.48 (-0.02) & 0.48 (-0.02) & 0.46 (-0.04) & 0.46 (-0.04) \\
\addlinespace
$\nu_1$ & -1.00 & -1.01 (-0.01) & -1.01 (-0.01) & -1.00 (-0.00) & -1.00 (+0.00) & -1.00 (-0.00) \\
$\nu_2$ & -0.50 & -0.54 (-0.04) & -0.54 (-0.04) & -0.54 (-0.04) & -0.54 (-0.04) & -0.54 (-0.04) \\
$\nu_3$ & 0.00 & 0.01 (+0.01) & 0.01 (+0.01) & 0.01 (+0.01) & 0.01 (+0.01) & 0.01 (+0.01) \\
$\nu_4$ & 0.50 & 0.51 (+0.01) & 0.51 (+0.01) & 0.51 (+0.01) & 0.51 (+0.01) & 0.51 (+0.01) \\
$\nu_5$ & 1.00 & 1.03 (+0.03) & 1.03 (+0.03) & 1.03 (+0.03) & 1.03 (+0.03) & 1.03 (+0.03) \\
\addlinespace
$\lambda_{1,2}$ & 1.02 & 1.04 (+0.02) & 1.04 (+0.02) & 1.04 (+0.02) & 1.05 (+0.03) & 1.04 (+0.02) \\
$\lambda_{1,3}$ & 0.78 & 0.80 (+0.02) & 0.80 (+0.02) & 0.80 (+0.02) & 0.80 (+0.02) & 0.80 (+0.02) \\
$\lambda_{1,4}$ & 1.06 & 1.11 (+0.06) & 1.11 (+0.05) & 1.11 (+0.06) & 1.12 (+0.06) & 1.11 (+0.05) \\
$\lambda_{1,5}$ & 0.82 & 0.89 (+0.07) & 0.89 (+0.07) & 0.89 (+0.07) & 0.89 (+0.07) & 0.89 (+0.06) \\
\addlinespace
$\lambda_{2,2}$ & 0.89 & 0.92 (+0.03) & 0.92 (+0.03) & 0.92 (+0.03) & 0.92 (+0.04) & 0.92 (+0.04) \\
$\lambda_{2,3}$ & 0.82 & 0.87 (+0.04) & 0.87 (+0.05) & 0.87 (+0.05) & 0.88 (+0.05) & 0.88 (+0.05) \\
$\lambda_{2,4}$ & 0.85 & 0.89 (+0.04) & 0.89 (+0.04) & 0.89 (+0.04) & 0.90 (+0.04) & 0.90 (+0.04) \\
$\lambda_{2,5}$ & 0.87 & 0.94 (+0.06) & 0.94 (+0.07) & 0.94 (+0.07) & 0.95 (+0.07) & 0.95 (+0.07) \\
\bottomrule
\end{tabular}
\end{table}

\begin{table}[htbp]
\centering
\caption{Posterior recovery with the logit link in the $N = 50$ and $T = 50$ case. Numbers are posterior means averaged across five runs, and numbers in parentheses are average deviations from the true value.}
\label{tab:example1-bias_logit_N50_T50}
\small
\begin{tabular}{lrrrr}
\toprule
Parameter & True & Hybrid (NumPyro) & NUTS (NumPyro) & NUTS (Stan) \\
\midrule
$\phi$ & 0.40 & 0.40 (+0.00) & 0.40 (+0.00) & 0.40 (+0.00) \\
$\psi_1$ & 0.84 & 0.82 (-0.02) & 0.82 (-0.02) & 0.82 (-0.02) \\
$\psi_2$ & 0.50 & 0.46 (-0.04) & 0.46 (-0.04) & 0.46 (-0.04) \\
\addlinespace
$\nu_1$ & -1.00 & -1.00 (-0.00) & -1.00 (-0.00) & -1.00 (+0.00) \\
$\nu_2$ & -0.50 & -0.51 (-0.01) & -0.51 (-0.01) & -0.51 (-0.01) \\
$\nu_3$ & 0.00 & 0.01 (+0.01) & 0.01 (+0.01) & 0.01 (+0.01) \\
$\nu_4$ & 0.50 & 0.49 (-0.01) & 0.49 (-0.01) & 0.49 (-0.01) \\
$\nu_5$ & 1.00 & 1.03 (+0.03) & 1.03 (+0.03) & 1.03 (+0.03) \\
\addlinespace
$\lambda_{1,2}$ & 1.02 & 1.03 (+0.01) & 1.04 (+0.01) & 1.04 (+0.01) \\
$\lambda_{1,3}$ & 0.78 & 0.81 (+0.03) & 0.81 (+0.03) & 0.81 (+0.03) \\
$\lambda_{1,4}$ & 1.06 & 1.11 (+0.06) & 1.11 (+0.06) & 1.12 (+0.06) \\
$\lambda_{1,5}$ & 0.82 & 0.86 (+0.04) & 0.86 (+0.04) & 0.86 (+0.04) \\
\addlinespace
$\lambda_{2,2}$ & 0.89 & 0.93 (+0.04) & 0.93 (+0.04) & 0.93 (+0.04) \\
$\lambda_{2,3}$ & 0.82 & 0.90 (+0.07) & 0.90 (+0.07) & 0.90 (+0.07) \\
$\lambda_{2,4}$ & 0.85 & 0.88 (+0.03) & 0.88 (+0.03) & 0.88 (+0.03) \\
$\lambda_{2,5}$ & 0.87 & 0.94 (+0.07) & 0.94 (+0.07) & 0.95 (+0.07) \\
\bottomrule
\end{tabular}
\end{table}

\subsection{Participant-Varying Dynamics}

Figure \ref{fig:example2-bulk-efficiency-probit} shows bulk efficiency with the probit link. Figures \ref{fig:example2-tail-efficiency-probit} and \ref{fig:example2-tail-efficiency-logit} show tail efficiency with the probit and logit links, respectively. Note that we only ran the two \textsc{NumPyro} implementations with the probit link. Tables \ref{tab:example2-efficiency_probit} and \ref{tab:example2-efficiency_logit} show the efficiency gain of the hybrid sampler over the other algorithms for all cases. Table \ref{tab:example2-iterations} shows the number of iterations run for each algorithm in each condition. Tables \ref{tab:example2-diag_probit} and \ref{tab:example2-diag_logit} show MCMC diagnostics for all settings. Tables \ref{tab:example2-bias_probit_N50_T50} and \ref{tab:example2-bias_logit_N50_T50} show the average posterior means together with biases for all parameters and algorithms in the $N=50$ and $T=50$ case, again suggesting that our implementations correctly target the posterior distribution. Similar tables for other $N$ and $T$ can be generated with the script \texttt{code/example2/make\_posterior\_tables.py} in the OSF repository, and they all give very similar numbers.

For all algorithms, the target acceptance probability during warm-up was set to $0.95$ and the maximum treedepth was $10$. The initial values were identical to the ones used with participant-invariant dynamics, except for $\mu_{\phi} = 0$, $\mu_{\psi}=0$, $\mu_{\lambda,j}=1$, $\omega_{\phi}^{2} = 0.3^{2}$, $\omega_{\psi}^{2} = 0.3^{2}$, and $\omega_{\lambda}^{2} = 0.3^{2}$.

\begin{figure}
    \includegraphics[width=\linewidth]{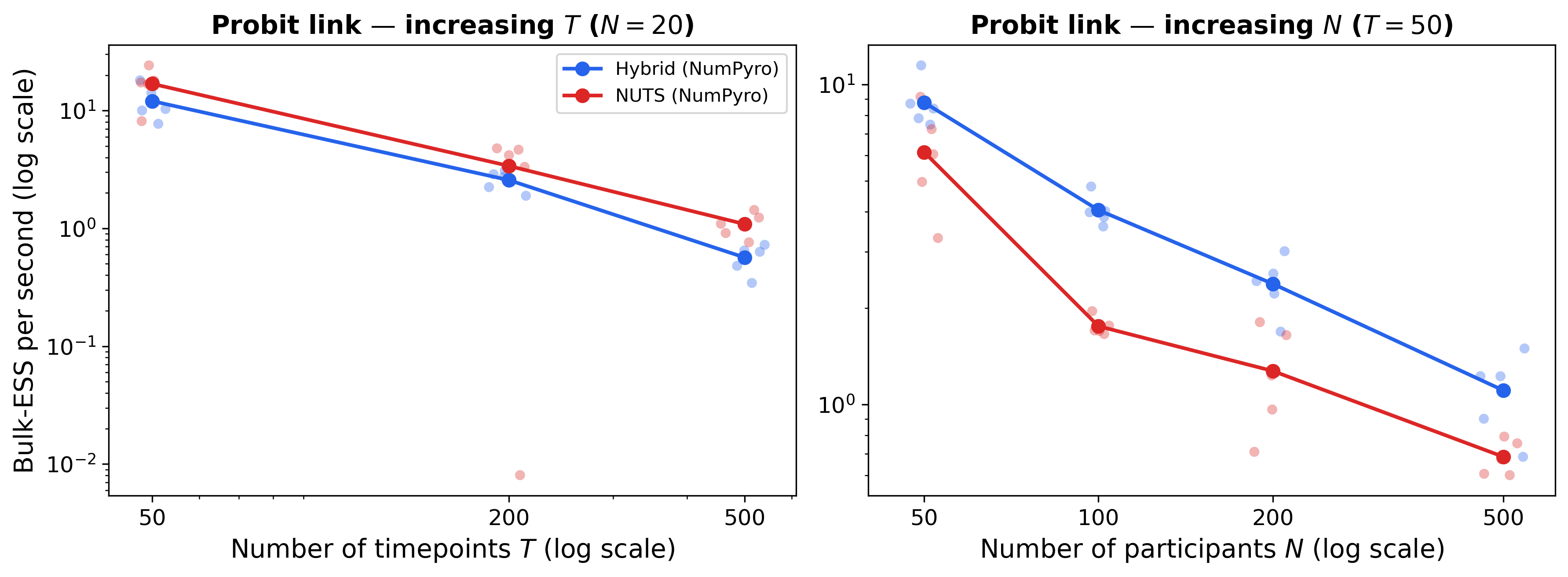}
    \caption{Bulk efficiency plot for the five-indicator AR(1) model with participant-varying dynamics with probit link.}
    \label{fig:example2-bulk-efficiency-probit}
\end{figure}

\begin{figure}
    \includegraphics[width=\linewidth]{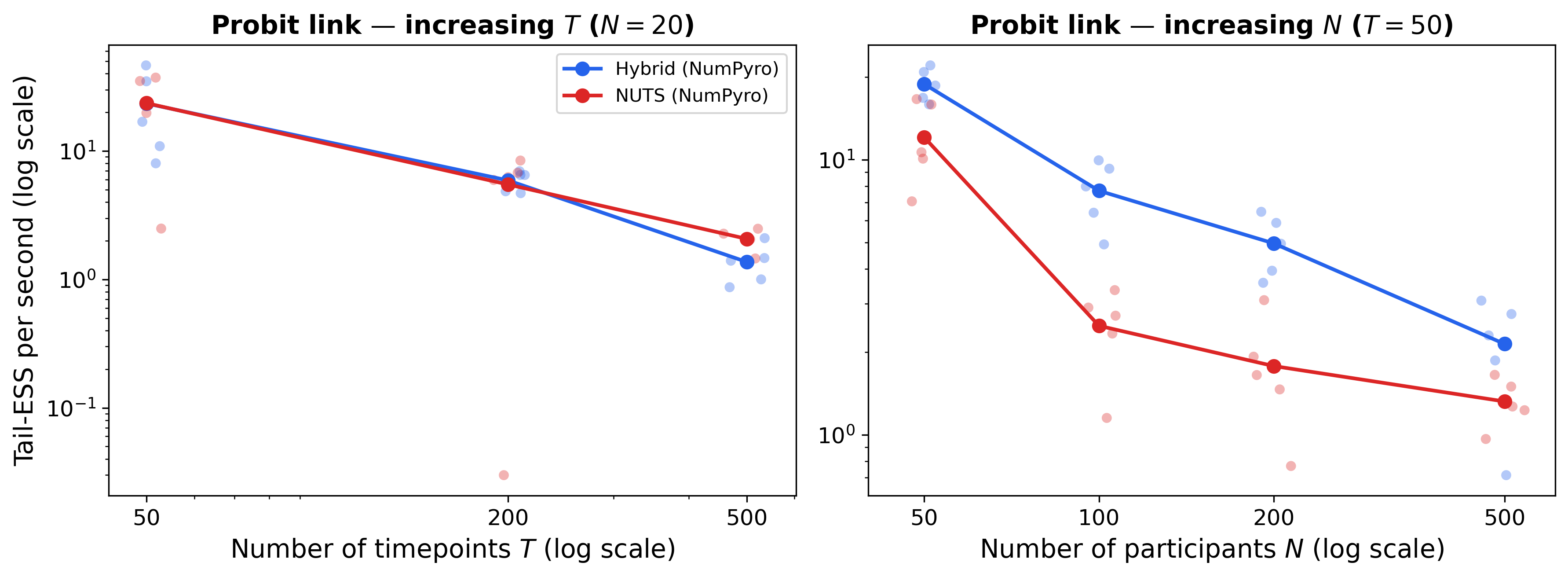}
    \caption{Tail efficiency plot for the five-indicator AR(1) model with participant-varying dynamics with probit link.}
    \label{fig:example2-tail-efficiency-probit}
\end{figure}

\begin{figure}
    \includegraphics[width=\linewidth]{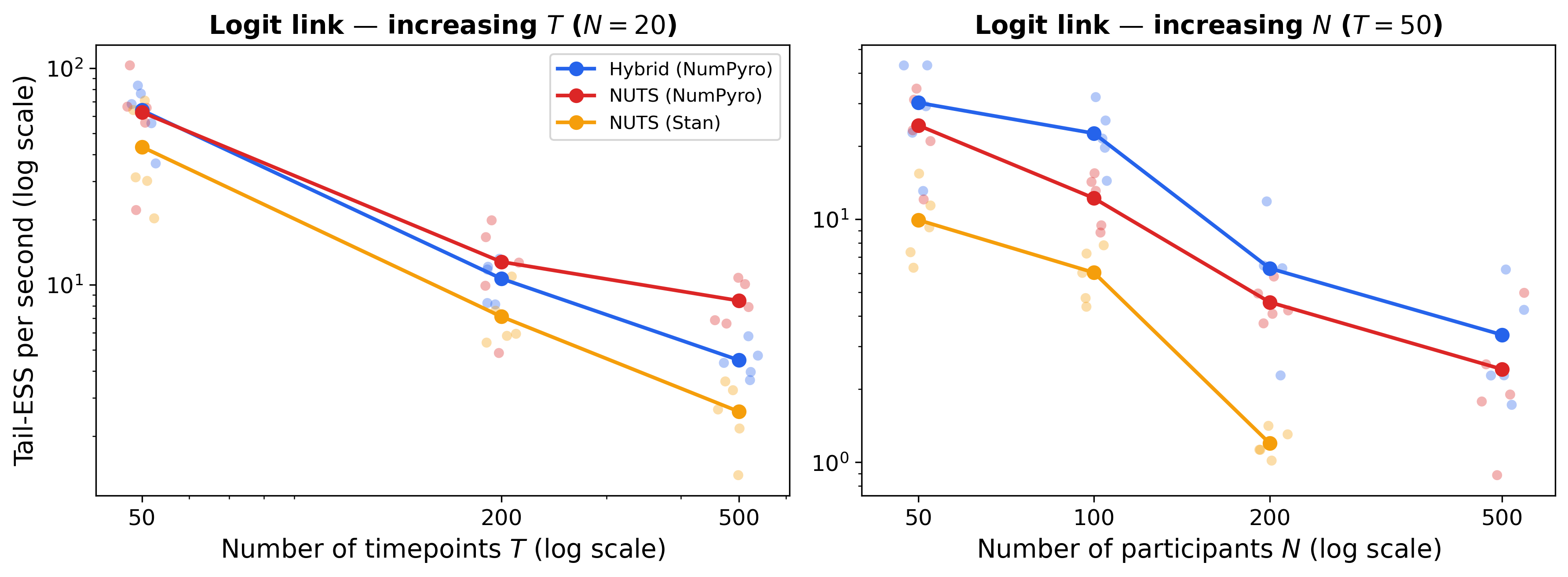}
    \caption{Tail efficiency plot for the five-indicator AR(1) model with participant-varying dynamics with logit link.}
    \label{fig:example2-tail-efficiency-logit}
\end{figure}

\begin{table}[htbp]
\centering
\caption{Efficiency gain of the hybrid sampler for the probit  link.}
\label{tab:example2-efficiency_probit}
\small
\begin{tabular}{ccrr}
\toprule
$N$ & $T$ & NUTS (NumPyro) & NUTS (Stan) \\
\midrule
20 & 50 & 0.7 & --- \\
20 & 200 & 0.8 & --- \\
20 & 500 & 0.5 & --- \\
50 & 50 & 1.4 & --- \\
100 & 50 & 2.3 & --- \\
200 & 50 & 1.9 & --- \\
500 & 50 & 1.6 & --- \\
\bottomrule
\end{tabular}
\end{table}

\begin{table}[htbp]
\centering
\caption{Efficiency gain of the hybrid sampler for the logit  link.}
\label{tab:example2-efficiency_logit}
\small
\begin{tabular}{ccrr}
\toprule
$N$ & $T$ & NUTS (NumPyro) & NUTS (Stan) \\
\midrule
20 & 50 & 0.9 & 1.3 \\
20 & 200 & 0.8 & 1.5 \\
20 & 500 & 0.5 & 1.7 \\
50 & 50 & 1.4 & 2.5 \\
100 & 50 & 1.7 & 3.9 \\
200 & 50 & 1.5 & 5.7 \\
500 & 50 & 1.6 & --- \\
\bottomrule
\end{tabular}
\end{table}

\begin{table}[htbp]
\centering
\caption{Number of iterations (warm-up / sampling) for all conditions and algorithms.}
\label{tab:example2-iterations}
\small
\begin{tabular}{llccc}
\toprule
$N$ & $T$ & Hybrid (NumPyro) & NUTS (NumPyro) & NUTS (Stan) \\
\midrule
20 & 50 & 1,000 / 4,000 & 1,000 / 4,000 & 1,000 / 4,000 \\
20 & 200 & 1,000 / 4,000 & 1,000 / 4,000 & 1,000 / 4,000 \\
20 & 500 & 1,000 / 4,000 & 1,000 / 4,000 & 1,000 / 4,000 \\
50 & 50 & 1,000 / 4,000 & 1,000 / 4,000 & 1,000 / 4,000 \\
100 & 50 & 1,000 / 4,000 & 1,000 / 4,000 & 1,000 / 4,000 \\
200 & 50 & 1,000 / 4,000 & 1,000 / 4,000 & 1,000 / 4,000 \\
500 & 50 & 1,000 / 4,000 & 1,000 / 4,000 & -- \\
\bottomrule
\end{tabular}
\end{table}

\begin{table}[htbp]
\centering
\caption{MCMC diagnostics for the probit link, averaged over 5 datasets.}
\label{tab:example2-diag_probit}
\small
\begin{tabular}{cclrrrrr}
\toprule
$N$ & $T$ & Algorithm & Wall time (s) & Min Bulk-ESS & Min Tail-ESS & Max $\hat{R}$ & Min Bulk-ESS/s \\
\midrule
20 & 50 & Hybrid (NumPyro) & 45 & 526 & 991 & 1.010 & 12 \\
 &  & NUTS (NumPyro) & 117 & 1869 & 2507 & 1.000 & 17 \\
\midrule
20 & 200 & Hybrid (NumPyro) & 288 & 723 & 1670 & 1.010 & 2.6 \\
 &  & NUTS (NumPyro) & 558 & 1638 & 2633 & 1.106 & 3.4 \\
\midrule
20 & 500 & Hybrid (NumPyro) & 1130 & 629 & 1512 & 1.008 & 0.57 \\
 &  & NUTS (NumPyro) & 1214 & 1319 & 2504 & 1.002 & 1.1 \\
\midrule
50 & 50 & Hybrid (NumPyro) & 60 & 515 & 1110 & 1.012 & 8.8 \\
 &  & NUTS (NumPyro) & 261 & 1478 & 2877 & 1.000 & 6.1 \\
\midrule
100 & 50 & Hybrid (NumPyro) & 112 & 454 & 864 & 1.010 & 4.0 \\
 &  & NUTS (NumPyro) & 640 & 1119 & 1579 & 1.002 & 1.8 \\
\midrule
200 & 50 & Hybrid (NumPyro) & 184 & 432 & 904 & 1.010 & 2.4 \\
 &  & NUTS (NumPyro) & 1264 & 1490 & 2016 & 1.002 & 1.3 \\
\midrule
500 & 50 & Hybrid (NumPyro) & 396 & 438 & 844 & 1.012 & 1.1 \\
 &  & NUTS (NumPyro) & 2647 & 1816 & 3496 & 1.000 & 0.69 \\
\bottomrule
\end{tabular}
\end{table}

\begin{table}[htbp]
\centering
\caption{MCMC diagnostics for the logit link, averaged over 5 datasets.}
\label{tab:example2-diag_logit}
\small
\begin{tabular}{cclrrrrr}
\toprule
$N$ & $T$ & Algorithm & Wall time (s) & Min Bulk-ESS & Min Tail-ESS & Max $\hat{R}$ & Min Bulk-ESS/s \\
\midrule
20 & 50 & Hybrid (NumPyro) & 57 & 1900 & 3500 & 1.002 & 35 \\
 &  & NUTS (NumPyro) & 51 & 1897 & 2948 & 1.000 & 40 \\
 &  & NUTS (Stan) & 75 & 1808 & 3019 & 1.000 & 26 \\
\midrule
20 & 200 & Hybrid (NumPyro) & 184 & 1182 & 1985 & 1.000 & 6.5 \\
 &  & NUTS (NumPyro) & 166 & 1328 & 1982 & 1.000 & 8.2 \\
 &  & NUTS (Stan) & 320 & 1333 & 2193 & 1.000 & 4.3 \\
\midrule
20 & 500 & Hybrid (NumPyro) & 631 & 1469 & 2831 & 1.002 & 2.3 \\
 &  & NUTS (NumPyro) & 416 & 1891 & 3516 & 1.000 & 4.5 \\
 &  & NUTS (Stan) & 1290 & 1833 & 3381 & 1.000 & 1.4 \\
\midrule
50 & 50 & Hybrid (NumPyro) & 74 & 1557 & 2072 & 1.002 & 22 \\
 &  & NUTS (NumPyro) & 98 & 1545 & 2383 & 1.000 & 16 \\
 &  & NUTS (Stan) & 160 & 1408 & 1587 & 1.004 & 8.8 \\
\midrule
100 & 50 & Hybrid (NumPyro) & 116 & 1470 & 2608 & 1.000 & 13 \\
 &  & NUTS (NumPyro) & 161 & 1243 & 1967 & 1.004 & 7.7 \\
 &  & NUTS (Stan) & 399 & 1316 & 2405 & 1.002 & 3.3 \\
\midrule
200 & 50 & Hybrid (NumPyro) & 274 & 1207 & 1621 & 1.004 & 4.8 \\
 &  & NUTS (NumPyro) & 271 & 836 & 1220 & 1.008 & 3.2 \\
 &  & NUTS (Stan) & 928 & 765 & 1100 & 1.004 & 0.84 \\
\midrule
500 & 50 & Hybrid (NumPyro) & 656 & 1461 & 2173 & 1.000 & 2.2 \\
 &  & NUTS (NumPyro) & 967 & 1297 & 2208 & 1.002 & 1.4 \\
\bottomrule
\end{tabular}
\end{table}

\begin{table}[htbp]
\centering
\caption{Posterior recovery with the probit link in the $N = 50$ and $T = 50$ case. Numbers are posterior means averaged across five runs, and numbers in parentheses are average deviations from the true value.}
\label{tab:example2-bias_probit_N50_T50}
\small
\begin{tabular}{lrrr}
\toprule
Parameter & True & Hybrid (NumPyro) & NUTS (NumPyro) \\
\midrule
$\mu_\phi$ & 0.42 & 0.42 (-0.00) & 0.42 (-0.00) \\
$\psi_2$ & 0.50 & 0.51 (+0.01) & 0.51 (+0.01) \\
\addlinespace
$\nu_1$ & -1.00 & -1.01 (-0.01) & -1.01 (-0.01) \\
$\nu_2$ & -0.50 & -0.53 (-0.03) & -0.53 (-0.03) \\
$\nu_3$ & 0.00 & -0.01 (-0.01) & -0.01 (-0.01) \\
$\nu_4$ & 0.50 & 0.50 (-0.00) & 0.50 (+0.00) \\
$\nu_5$ & 1.00 & 1.01 (+0.01) & 1.02 (+0.02) \\
\addlinespace
$\mu_{\lambda_{1,2}}$ & 1.02 & 1.07 (+0.05) & 1.08 (+0.05) \\
$\mu_{\lambda_{1,3}}$ & 0.78 & 0.80 (+0.02) & 0.80 (+0.02) \\
$\mu_{\lambda_{1,4}}$ & 1.06 & 1.13 (+0.07) & 1.13 (+0.07) \\
$\mu_{\lambda_{1,5}}$ & 0.82 & 0.88 (+0.06) & 0.89 (+0.06) \\
\addlinespace
$\lambda_{2,2}$ & 0.89 & 0.92 (+0.03) & 0.92 (+0.03) \\
$\lambda_{2,3}$ & 0.82 & 0.87 (+0.05) & 0.88 (+0.05) \\
$\lambda_{2,4}$ & 0.85 & 0.90 (+0.04) & 0.90 (+0.05) \\
$\lambda_{2,5}$ & 0.87 & 0.91 (+0.03) & 0.91 (+0.03) \\
\bottomrule
\end{tabular}
\end{table}

\begin{table}[htbp]
\centering
\caption{Posterior recovery with the logit link in the $N = 50$ and $T = 50$ case. Numbers are posterior means averaged across five runs, and numbers in parentheses are average deviations from the true value.}
\label{tab:example2-bias_logit_N50_T50}
\small
\begin{tabular}{lrrrr}
\toprule
Parameter & True & Hybrid (NumPyro) & NUTS (NumPyro) & NUTS (Stan) \\
\midrule
$\mu_\phi$ & 0.42 & 0.43 (+0.01) & 0.43 (+0.01) & 0.43 (+0.01) \\
$\psi_2$ & 0.50 & 0.48 (-0.02) & 0.48 (-0.02) & 0.48 (-0.02) \\
\addlinespace
$\nu_1$ & -1.00 & -1.01 (-0.01) & -1.01 (-0.01) & -1.01 (-0.01) \\
$\nu_2$ & -0.50 & -0.52 (-0.02) & -0.52 (-0.02) & -0.52 (-0.02) \\
$\nu_3$ & 0.00 & -0.01 (-0.01) & -0.01 (-0.01) & -0.01 (-0.01) \\
$\nu_4$ & 0.50 & 0.50 (-0.00) & 0.50 (-0.00) & 0.49 (-0.01) \\
$\nu_5$ & 1.00 & 0.99 (-0.01) & 0.99 (-0.01) & 0.99 (-0.01) \\
\addlinespace
$\mu_{\lambda_{1,2}}$ & 1.02 & 1.11 (+0.08) & 1.10 (+0.08) & 1.10 (+0.08) \\
$\mu_{\lambda_{1,3}}$ & 0.78 & 0.86 (+0.07) & 0.86 (+0.07) & 0.86 (+0.07) \\
$\mu_{\lambda_{1,4}}$ & 1.06 & 1.11 (+0.05) & 1.11 (+0.05) & 1.11 (+0.05) \\
$\mu_{\lambda_{1,5}}$ & 0.82 & 0.85 (+0.03) & 0.85 (+0.03) & 0.85 (+0.03) \\
\addlinespace
$\lambda_{2,2}$ & 0.89 & 0.94 (+0.05) & 0.94 (+0.05) & 0.93 (+0.05) \\
$\lambda_{2,3}$ & 0.82 & 0.88 (+0.05) & 0.88 (+0.05) & 0.87 (+0.05) \\
$\lambda_{2,4}$ & 0.85 & 0.89 (+0.04) & 0.89 (+0.04) & 0.89 (+0.04) \\
$\lambda_{2,5}$ & 0.87 & 0.89 (+0.01) & 0.89 (+0.02) & 0.89 (+0.01) \\
\bottomrule
\end{tabular}
\end{table}

\section{Multivariate Binomial Response Nine-Indicator VAR(1) Model}

The prior distributions for this model were as follows. All elements of $\bm{\mu}_{\bm{\Phi}}$ were sampled independently from a standard normal distribution and the between-participant standard deviations around these means had prior $\mathcal{TN}_{[0,\infty]}(0, 0.5^{2})$. The population means of the diagonal elements of the lower-triangular Cholesky factors of the within-level process noise were log-normally distributed with location 0 and scale 1 and the off-diagonals were normally distributed with mean zero and variance $0.5^2$. The standard deviations of the logarithm of individual diagonal elements $\bm{L}_{1,i}$ around these population means had prior $\mathcal{TN}_{[0,\infty]}(0, 0.5^{2})$, and standard deviations of the off-diagonal elements had the same prior. The population means for all factor loadings had prior $\mathcal{N}(1.0, 0.5^{2})$ and for $\bm{\Lambda}_{1,i}$ the standard deviation of the individual factor loadings around these means had prior $\mathcal{TN}_{[0, \infty]}(0, 1)$. The log-diagonals of the Cholesky factors of $\bm{\Psi}_{2}$ had $\mathcal{N}(0,1)$ and the off-diagonal elements had prior $\mathcal{N}(0, 0.5^2)$. The intercepts had prior $\nu_{j} \sim \mathcal{N}(0, 4)$ for $j=1, \dots,9$. 

Since this model has a very large number of parameters, making tables comparing posterior means would be impractical. Instead, Figure \ref{fig:example4_posteriors} shows the average posterior mean across the five seeds for the two algorithms, of all population-level parameters. The fact that they all lie close to the diagonal, suggests that the implementations were correct.

\begin{figure}
    \includegraphics[width=\linewidth]{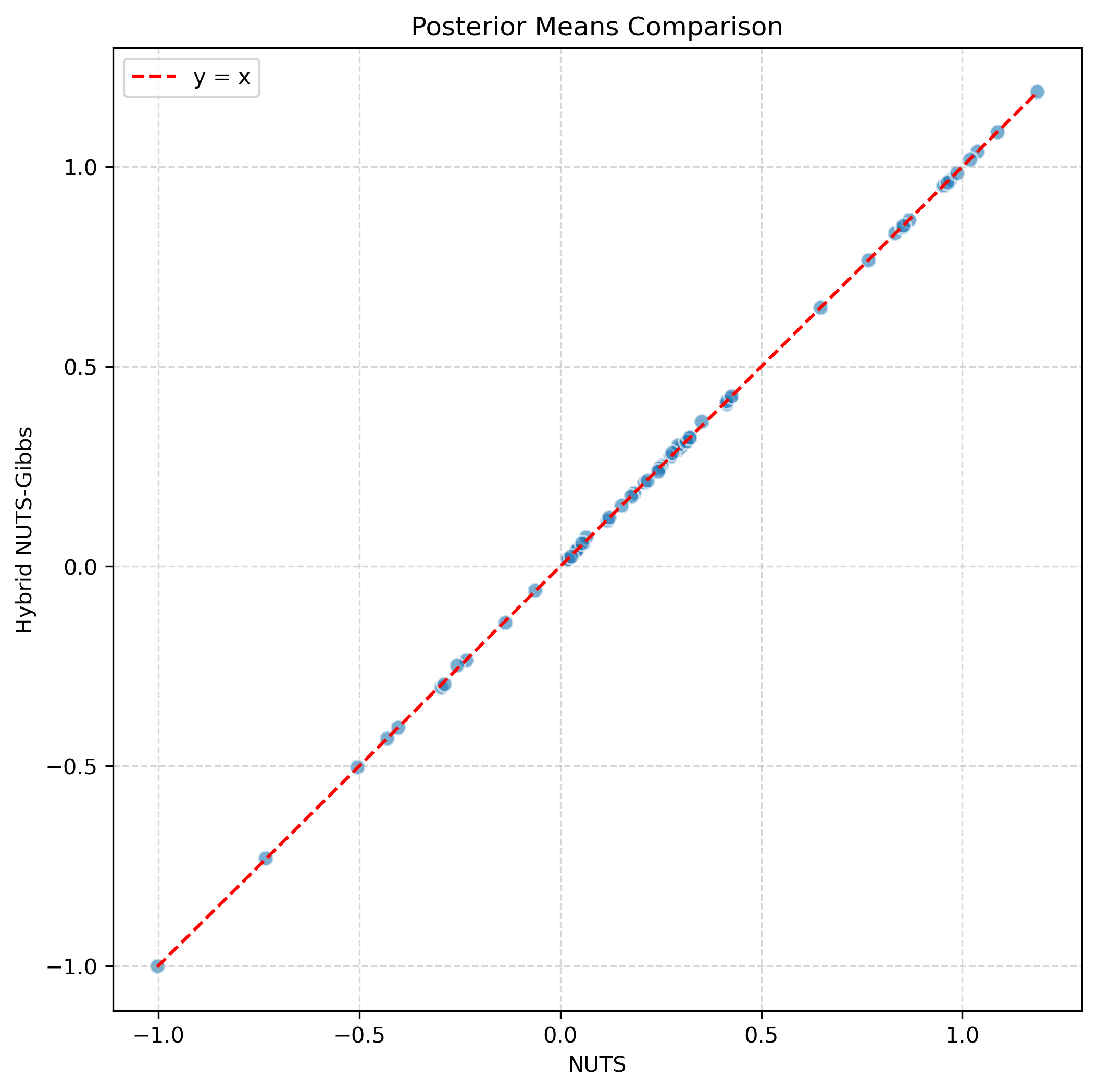}
    \caption{Comparison of parameter posterior means between the NUTS and Hybrid NUTS-Gibbs samplers for the VAR(1) model.}
    \label{fig:example4_posteriors}
\end{figure}

\end{document}